\documentclass[aos]{imsart}

\RequirePackage[OT1]{fontenc}
\usepackage{amsthm,amsmath,natbib}
\RequirePackage[colorlinks,citecolor=blue,urlcolor=blue]{hyperref}

\arxiv{arXiv:1606.03059}

\usepackage{latexsym,amsbsy,amssymb}
\usepackage{url}
\usepackage{color}
\usepackage{subfigure}
\usepackage{graphicx}
\usepackage{caption}

\newtheorem{Theorem}{Theorem}[section]
\newtheorem{Definition}{Definition}[section]

\newtheorem{Corollary}{Corollary}[section]

\newtheorem{Assumption}{Assumption}[section]

\newtheorem{Lemma}{Lemma}[section]

\numberwithin{equation}{section}

\newcommand{\h}{\hspace*{.24in}}
\newcommand{\mb}{\mathbf}
\newcommand{\mbb}{\mathbb}
\newcommand{\mc}{\mathcal}

\newcommand{\argmin}{\operatornamewithlimits{argmin}}
\newcommand{\op}{\operatorname}

\startlocaldefs
\numberwithin{equation}{section}
\theoremstyle{plain}

\endlocaldefs

\begin{document}

\begin{frontmatter}

\title{Consistency and convergence rate of phylogenetic inference via regularization}
\runtitle{phylogenetic inference via regularization}
\thankstext{T1}{These authors contributed equally to this work.}

\begin{aug}
\author{\fnms{Vu} \snm{Dinh}\thanksref{T1,m1}\ead[label=e1]{vdinh@fredhutch.org}},
\author{\fnms{Lam Si Tung} \snm{Ho}\thanksref{T1,m2}\ead[label=e2]{lamho@ucla.edu}}
\author{\fnms{Marc A.} \snm{Suchard}\thanksref{m2}\ead[label=e3]{msuchard@ucla.edu}}
\and
\author{\fnms{Frederick A.} \snm{Matsen IV}\thanksref{m1}
\ead[label=e4]{matsen@fredhutch.org}}

\runauthor{V. Dinh, Lam Ho, Marc Suchard and Frederick Matsen IV}

\affiliation{Fred Hutchinson Cancer Research Center\thanksmark{m1} and University of California, Los Angeles\thanksmark{m2}}

\address{Vu Dinh\\
Program in Computational Biology\\
Fred Hutchinson Cancer Research Center\\
\printead{e1}}

\address{Lam Si Tung Ho \\
Department of Biostatistics \\ 
University of California, Los Angeles \\
\printead{e2}}

\address{Marc A.~Suchard \\
Departments of Biomathematics, Biostatistics and Human Genetics \\
University of California, Los Angeles\\
\printead{e3}}

\address{Frederick A.~Matsen IV\\
Program in Computational Biology\\
Fred Hutchinson Cancer Research Center\\
\printead{e4}}

\end{aug}

\begin{abstract}
It is common in phylogenetics to have some, perhaps partial, information about the overall evolutionary tree of a group of organisms and wish to find an evolutionary tree of a specific gene for those organisms.
There may not be enough information in the gene sequences alone to accurately reconstruct the correct ``gene tree.''
Although the gene tree may deviate from the ``species tree'' due to a variety of genetic processes, in the absence of evidence to the contrary it is parsimonious to assume that they agree.
A common statistical approach in these situations is to develop a likelihood penalty to incorporate such additional information.
Recent studies using simulation and empirical data suggest that a likelihood penalty quantifying concordance with a species tree can significantly improve the accuracy of gene tree reconstruction compared to using sequence data alone.
However, the consistency of such an approach has not yet been established, nor have convergence rates been bounded.
Because phylogenetics is a non-standard inference problem, the standard theory does not apply.
In this paper, we propose a penalized maximum likelihood estimator for gene tree reconstruction, where the penalty is the square of the Billera-Holmes-Vogtmann geodesic distance from the gene tree to the species tree.
We prove that this method is consistent, and derive its convergence rate for estimating the discrete gene tree structure and continuous edge lengths (representing the amount of evolution that has occurred on that branch) simultaneously.
We find that the regularized estimator is ``adaptive fast converging,'' meaning that it can reconstruct all edges of length greater than any given threshold from gene sequences of polynomial length.
Our method does not require the species tree to be known exactly; in fact, our asymptotic theory holds for any such guide tree.
\end{abstract}

\begin{keyword}[class=MSC]
\kwd[Primary ]{05C05}
\kwd{62F12}
\kwd[; secondary ]{92B10}
\kwd{92D15}
\end{keyword}

\begin{keyword}
\kwd{phylogenetics}
\kwd{tree reconstruction}
\kwd{gene tree}
\kwd{species tree}
\kwd{maximum likelihood estimator}
\kwd{regularization}
\end{keyword}

\end{frontmatter}

\section{Introduction}

Molecular phylogenetics is the reconstruction of evolutionary history from molecular sequences, typically DNA sampled from the present day.
One common means of performing molecular phylogenetics is likelihood-based ``gene tree analysis,'' in which the molecular sequence of a gene is assumed to evolve according to a continuous time Markov chain (CTMC) along the branches of an unknown phylogenetic tree.
The likelihood function corresponding to this CTMC is then used as the sole criterion to choose the ``best'' tree.

However, sometimes there is insufficient signal in the gene sequences to deliver a confident estimate.
This may happen because of strong genetic conservation, insufficient evolutionary time, short sequences, or other processes obscuring the tree signal.
Insufficient signal, in turn, results in an error-prone high variance estimator in which small modifications of the input data can result in rather different inferred trees.

The introduction of an additional penalty parameter in the optimality function through regularization is a common statistical response to such ill-posed problems.
This penalty parameter can introduce extra information into the problem, resulting in a lower variance estimator which in many cases can be shown to be an improvement over the raw estimator.
Regularization-type strategies have made a recent appearance in phylogenetics, via the fact that gene trees evolve within an organism-level species tree.
The gene tree history may differ from a species tree due to processes such as gene duplications, losses, and horizontal gene transfer, however despite these processes there is significant signal bringing the various gene trees together to their shared species tree \citep{boussau2010genomes}.

Phylogenetic regularization has been implemented using two approaches thus far.
The model-based approach to phylogenetic regularization works to model the entire process of gene and species tree development, resulting in a comprehensive joint likelihood function for an ensemble of gene trees, either conditioning on a species tree or including it as part of the joint likelihood function \citep{liu2007species,aakerborg2009simultaneous,heled2010bayesian,rasmussen2011bayesian,boussau2013genome, Mahmudi2013-lp, szollHosi2015inference}.
Consistency of such approaches is a matter of continuing research \citep{roch2015robustness}.
Another approach is to use a penalty function describing the level of divergence of a given gene tree from a pre-specified species tree \citep{david2011rapid,wu2013treefix,bansal2014improved,scornavacca2014joint}, an approach more typical of the regularization approaches common in statistics and machine learning.
A relatively simple penalization approach can give performance competitive with inference under a full probabilistic model at a substantially smaller computational cost \citep{wu2013treefix}.

In another direction, probabilists have built a powerful theory showing consistency of, and describing sequence length requirements for, accurate phylogenetic inference.
For individual gene tree inference without additional information, this has recently reached a high-water mark with the demonstration that accurate maximum likelihood (ML) inference can be done with sequences of length polynomial in the number of taxa \citep{roch2015phase}, which is asymptotically tight given previous results \citep{mossel2003impossibility}.

However, we are not aware of any work giving proofs of consistency for penalty-based methods of gene tree inference, or more generally giving sequence length requirements for gene tree inference in the presence of a species tree.
In addition, the work that has been done showing polynomial sequence length requirements for ML inference has reduced the inference problem to a purely combinatorial one by discretizing branch lengths rather than treating branch lengths as continuous parameters.

In this paper, we propose a regularization framework for gene tree reconstruction and prove consistency of our estimation method as sequence length increases to infinity. 
Furthermore, we provide the corresponding sequence length requirements for accurate estimation.
The convergence theory is simultaneously for continuous branch lengths and tree topology.
Our penalty is in terms of distance from a ``guide tree'', often taken as, and in this paper called, the species tree, although we make no assumptions about how the gene tree was generated from the guide tree.
In fact, the guide tree can be arbitrarily chosen, although of course better gene tree reconstructions can be expected from species trees closer to the correct gene tree.

Our penalty is in terms of geodesic distance in the Billera-Holmes-Vogtmann (BHV) space, which represents the discrete and continuous aspects of phylogenetic trees as an ensemble of orthants glued together along their edges \citep{billera2001geometry}.
The discrete changes induced by the geodesic distance are nearest-neighbor-interchanges between subtrees separated by zero length branches, modeling horizontal transfer, lineage sorting or deep coalescence, and gene duplication/extinction \citep{maddison1997gene}.
The BHV distance can be computed in polynomial time \citep{owen2011fast}, in contrast with other distances on tree spaces which are typically NP-hard to compute.
BHV space has been previously used to extend classical statistical methods for phylogenetic inference.
For example, \citet{nye2011principal} proposes an algorithm to identify principal paths in BHV space analogously to the standard Principal Component Analysis.

Using this framework, we are able to prove that the regularized ML estimator is consistent, and derive its \emph{global error} in various settings.
Our results show a polynomial sequence length requirement to recover the true topology, which is the best possible sample complexity of any tree reconstruction methods.
Moreover, unlike most of the previous algorithms that perform poorly on trees with indistinguishable edges, the regularized estimator can reconstruct all edges of length greater than any given threshold.
We can also derive a confidence region for the estimator and use that to assess the support of tree splits from data.

The paper is organized as follows. Section \ref{sec:math} introduces our mathematical framework including the BHV space of phylogenetic trees and the regularized maximum likelihood estimator for tree reconstruction using geodesic BHV distance.
We investigate properties of the geodesic distance in the BHV space and the likelihood in Section \ref{sec:geodesic} and Section \ref{sec:likelihood}, respectively.
We prove that the regularized maximum likelihood estimator is consistent and obtain its convergence rate in Section \ref{sec:asymptotic}.
Section \ref{sec:complexity} derives an explicit bound of the convergence rate and provides the sample complexity of our method to recover the true topology under the Jukes-Cantor model.
We apply our method to reconstruct gene tree for eight yeast species in Section \ref{sec:application}.
Finally, Section \ref{sec:discussion} discusses our results in the context of related work.


\section{Mathematical framework}
\label{sec:math}

\subsection{Phylogenetic tree}
Throughout this paper, the term \emph{phylogenetic tree} refers to a tree $T$ with leaves labeled by a set of taxon (i.e.\ organism) names.
Phylogeneticists often call these trees \emph{unrooted} to highlight the undirected nature of the edges.
Let $E(T)$ denote the edges of $T$ and $V(T)$ denote the vertices.
Any edge adjacent to a leaf is called an pendant edge, and any other edge is called an internal edge. Each edge $e$ will be associated with a non negative number $w_e$ called the branch length.
We define the \emph{diameter} of $T$ as $\max_{u,v \in V(T)} \sum_{e \in \op{path}(u,v)} w_e$. A tree is said to be resolved if it is bifurcating and all branch lengths are positive.

The topological distance $d_T(u, v)$ between vertices $u$ and $v$ in a tree is the number of edges in the path between vertices $u$ and $v$.
For an edge $e$ of $T$, let $T_1$ and $T_2$ be the two rooted subtrees of T obtained by deleting edge $e$ from $T$, and for $i = 1,2$, let $d_i(e)$ be the topological distance from the root of $T_i$ to its nearest leaf in $T_i$. The \emph{depth} of $T$ is defined as $\max_e \max \{d_1(e),d_2(e)\}$, where $e$ ranges over all internal edges in $T$.
Given an unrooted phylogenetic tree $T$ on a finite set $X$ of taxa, any subset $Y$ of $X$ induces a phylogenetic tree on taxon set $Y$, denoted $T |_Y$, which, is the subtree of $T$ that connects the taxa in $Y$ only.

We will also follow the standard independent and identically distributed (IID) setting for likelihood-based phylogenetics with a finite number of sites \citep[see][for more details]{felsenstein2004inferring}.
Let $\mc{S}$ denote the set of possible molecular sequence character states and let $r= |\mc{S}|$; for convenience, we assume that the states have indices $1$ to $r$.
We assume that mutation events occur according to a continuous-time Markov chain (CTMC) on states $\mc{S}$.
Specifically, the probability of ending in state $y$ after time $t$ given that the site started in state $x$ is given by the $xy$-th entry of $P(t)$, where $P(t)$ is the matrix valued function $P(t)=e^{Qt}$, and the matrix $Q$ is the instantaneous rate matrix of the CTMC evolutionary model.
The branch lengths represent the times $t$ during which the mutation process operates.
We assume that the rate matrix $Q$ is reversible with respect to a stationary distribution $\pi$ on the set of states $\mc{S}$.

To explore the  set of all trees with a given number of taxa and to describe trees that are ``near" to each other, we use the class of nearest neighbor interchange (NNI) moves \citep{robinson1971comparison}.
An NNI move is defined as a transformation that collapses an interior edge to zero and then expands the resulting degree 4 vertex into an edge and two degree 3 vertices in a new way (see Figure \ref{fig:NNI}).
Two trees $\tau_1$ and $\tau_2$ are said to be NNI-adjacent if there exists a single NNI move that transform $\tau_1$ into $\tau_2$.

\begin{figure}
\centering
  \includegraphics[width=0.7\linewidth]{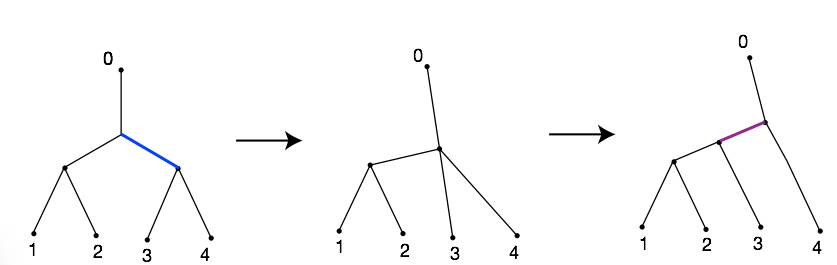}
  \caption{An NNI move collapses an interior edge to zero and then expands the resulting degree 4 vertex into an edge and two degree 3 vertices in a new way. }
  \label{fig:NNI}
\end{figure}

\subsection{BHV space}
The BHV space employs a cubical complex as the geometric model of tree space $\mathcal{T}$ on $n$ taxa as follows:
\begin{enumerate}
\item[1. ] $\mc{T}$ consists of a collection of orthants, each isomorphic to $\mathbb{R}_{\ge 0}^{2n-3}$.
\item[2. ] Each orthant itself corresponds uniquely to a tree topology, and the coordinates in each orthant parameterize the branch lengths for the corresponding tree.
\item[3. ] The adjacent orthants of the complex with the same dimension correspond to NNI-adjacent trees.
\end{enumerate}

\begin{figure}
\centering
  \includegraphics[width=0.5\linewidth]{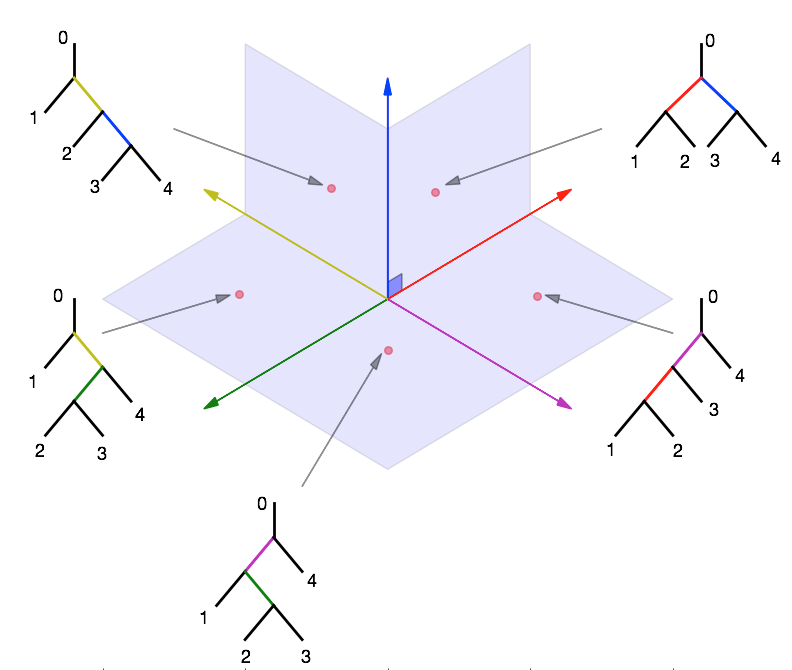}
  \caption{The BHV space is a cubical complex, where each orthant corresponds uniquely to a tree topology and the coordinates in each orthant parameterize its branch lengths. Adjacent orthants represent trees that can be transformed to each other via a single NNI move.}
  \label{fig:BHV}
\end{figure}

A simple visualization of a part of the BHV space of all trees with 5 taxa is provided in Figure \ref{fig:BHV}.
The BHV space is not a manifold, but it is equipped with a natural metric distance: the shortest path lying in the BHV space between the points. If two points lie in the same orthant, this distance is the usual Euclidean distance.
If two points are in different orthants, they can be joined by a sequence of straight segments, with each segment lying in a single orthant.
We can then measure the length of the path by adding up the lengths of the segments.
The distance between the two points is defined as the minimum of the lengths of such segmented paths joining the two points.

A segmented path giving the smallest distance between two points is called a \emph{geodesic}.
The geodesic connecting any two points in BHV space is unique and can be computed in polynomial time \citep{owen2011fast}.
The space is clearly locally compact.

In this paper, beside the BHV geodesic distance itself, we will also consider the branch-score (BS) distance.
This distance was proposed by \citet{Kuhner1994-tz} and shown by \citet{amenta2007approximating} to be an approximation to the BHV distance.
We say that a bipartition or \emph{split} of $V(T)$ (into two disjoint subset $A, B$) is in the tree $T$ if there is some edge $g \in E(T)$ such that all elements of $A$ lie on one side of $g$ and all elements of $B$ lie on the other side.
The branch-score distance between two trees $T_1$ and $T_2$ is defined as $d_{BS}(T_1,T_2) = \sqrt{\sum_{e}{(w_1(e)-w_2(e))^2}}$ where $w_i(e)$ is the length of the corresponding edge if the split $e$ is in the tree $T_i$, and $w_i(e)=0$ otherwise \citep{Kuhner1994-tz}.
The branch-score distance is equivalent to the BHV distance
$
d_{BS}(T_1,T_2) \le d_{BHV}(T_1,T_2) \le \sqrt{2} d_{BS}(T_1,T_2), ~\forall T_1, T_2
$
and can be computed in $O(n)$ time \citep{amenta2007approximating}.


\subsection{Phylogenetic likelihood and the forward operator}
For a fixed topology $\tau$, vector of branch lengths $q$ and observed sequences $\mb{Y}_k = (Y_1, Y_2,...,Y_k)$ in $\mc{S}^{n \times k}$ of length $k$ over $n$ taxa, the likelihood of observing $\mb{Y}_k$ given $\tau$ and $q$ has the form
\[
L_k(\tau, q) = \prod_{s=1}^k{\sum_{a^{(s)} }{ \left ( \prod_{(u,v)\in E(\tau, q)}{P^{uv}_{a^{(s)}_u a^{(s)}_v}( q_{uv})} \right ) \pi(a^{(s)}_{u_0})}}
\]
where $a^{(s)}$ ranges over all extensions of $Y_s$ to the internal nodes of the tree, $a^{(s)}_u$ denotes the assigned state of node $u$ by $a^{(s)}$, $E(\tau, q)$ denotes the set of tree edges, $P^{uv}$ is the probability transition matrix on edge $(u,v)$, and $u_0$ is an arbitrary internal node.

We will also denote $\ell_k = \log(L_k)$ and refer to it as the log-likelihood function given the observed sequence. It is known that $\ell_k$ is continuous on $\mathcal{T}$ and is smooth (up to the boundary) on each orthant of $\mathcal{T}$.
However, there is no general notion of differentiability of $\ell_k$ on the whole tree space.

Each tree $(\tau,q)$ generates a joint distribution on the single-site patterns at the leaf nodes, hereafter denoted by $P_{\tau,q}$. Define the \emph{forward operator} 
\begin{eqnarray*}
F: \h\h \mathcal{T}&\rightarrow& \mathbb{X}\\
 (\tau, q)&\mapsto&  P_{\tau,q}.
\end{eqnarray*}
where $\mathbb{X}$ denotes the space of all possible single-site distributions on the leaves.
Throughout the paper, we will use the notation $\operatorname{KL}(P_1, P_2)$ for the Kullback-Leibler divergence of  $P_2 \in \mathbb{X}$ from $P_1 \in \mathbb{X}$ .
We note that the likelihood function can be rewritten in term of the forward operator as
\[
L_k(\tau, q) = \prod_{s=1}^k{P_{\tau,q}(Y_s)}.
\]

In this paper, we mostly focus on the case when the sequence length $k$ increases while the number of taxa $n$ is fixed, with the exception of Section \ref{sec:complexity} where we quantify the distance between our estimate and the true tree in both $k$ and $n$.
We will make the following assumptions:

\begin{Assumption}[Identifiability] 
$P_{\tau,q} = P_{\tau',q'}  \Leftrightarrow  \tau=\tau'  ~ \text{and} ~~ q=q'.$
\label{assump:iden}
\end{Assumption}

\begin{Assumption}
We assume that the data $\mb{Y}_k$ are generated from a true tree $(\tau^*,q^*)$, a resolved tree with bounded branch lengths: let $e, f > 0$ be the lower bound of branch lengths of pendant edges and inner edges respectively, and $g > 0$ be the upper bound of all branch lengths.
Note that $e,f$, and $g$ may depend on the number of taxa.
We further assume that $e$ and $g$ are known.
\label{assump:leafedges}
\end{Assumption}

From now on, we denote $\mc{T}_{e,g}$ be the set of all trees for which edge lengths are bounded from above by $g$ and pendant edge lengths are bounded from below by $e$.


\subsection{Phylogenetic inference via regularization}
Given a species tree $(\tau^{\circ}, q^{\circ})$ and regularization parameter $\alpha_k > 0$, the regularized estimator $(\hat \tau_k, \hat q_k)$ of the tree is defined as the minimizer of the BHV-penalized phylogenetic likelihood function:
\begin{equation}
(\hat \tau_k, \hat q_k) :=  \argmin_{(\tau, q) \in \mathcal{T}}{- \frac{1}{k}\ell_k(\tau,q) +  \alpha_k R(\tau,q)}.
\label{eq1}
\end{equation}
Here $R(\tau,q) = d((\tau,q), (\tau^{\circ}, q^{\circ}))^2$ where $d(s,s')$ denotes the geodesic distance between the trees $s$ and $s'$ on the BHV space \citep{billera2001geometry}.
This formulation is analogous to the ridge estimator in the standard linear regression setting \citep{hoerl1962ridge}.
The existence of a minimizer as in $\eqref{eq1}$ is guaranteed by the following Lemma:
\begin{Lemma}[Existence of the estimator]
For fixed $\beta>0$ and observed sequence $\mb{Y}_k$, there is a $(\tau, q)$ minimizing $Z_{\beta, \mb{Y}_k}(\tau, q) =- \ell_k(\tau,q)  + \beta~R(\tau,q).$
\end{Lemma}

\begin{proof}
Let $(\tau_n, q_n)$ be a sequence such that $Z_{\beta, \mb{Y}_k}(\tau_n, q_n)$ converges to $\inf_{(\tau, q)} Z_{\beta, \mb{Y}_k} (\tau, q)$.
Note that $d((\tau_n,q_n), (\tau^{\circ},q^{\circ}))$ is bounded from above by Assumption \ref{assump:leafedges}, so the sequence lies inside a compact set. We deduce that a subsequence $(\tau_m, q_m)$ converges to some $(\tau^*,q^*) \in \mathcal{T}$.
Since the likelihood and the penalty $R$ are continuous, we deduce that $(\tau^*,q^*)$ is a minimizer of $Z_{\beta, \mb{Y}_k}$.
\end{proof}



\section{Properties of the regularization penalty}
\label{sec:geodesic}

In this section, we investigate the analytical properties of the penalty in our regularization problem.
We recall the definition of a strongly convex distance on a geodesic space.
\begin{Definition}[Strongly convex distance]
A distance $d$ on a space $\mc{X}$ is strongly convex if for any geodesic $\gamma: [0,1] \to \mathcal{X}$, $x \in \mathcal{X}$ and $t \in [0,1]$, we have $d(x, \gamma(t))^2 \le (1-t) d(x, \gamma(0))^2 + t d(x, \gamma(1))^2 - t(1-t) d (\gamma(0), \gamma(1))^2.$
\label{def:convex}
\end{Definition}

We have the following lemma about the convexity of the aforementioned distances:

\begin{Lemma}
The BHV geodesic distance is strongly convex on the whole tree space.
\label{lem:convex}
\end{Lemma}

\begin{proof}
It is known that the BHV space with its distance is a CAT($0$) space \citep{billera2001geometry}.
Standard results about distance on CAT($0$) spaces \citep[for example, see][]{bavcak2013proximal, bacak2014computing} imply strong convexity of the geodesic distance.
\end{proof}

Recall that in our formulation, $ R(\tau,q) = d((\tau,q), (\tau^{\circ},q^{\circ}))^2$.
Letting $x^{\circ} := (\tau^{\circ}, q^{\circ})$ denote our species tree with branch lengths, we note that $R(x)$ is locally Lipschitz by two applications of the triangle inequality:
\begin{equation}
\begin{aligned}
|R(y) - R(x)| &= \left | d(y,x^{\circ})-d(x,x^{\circ}) \right |  \left ( d(y,x^{\circ})+d(x,x^{\circ}) \right ) \\
& \le   d(y,x) \left ( d(y,x)+ 2 d(x, x^{\circ}) \right ).
\label{eqn:LipschitzR}
\end{aligned}
\end{equation}
We will use the notation $\gamma_{x,y}(t),~t \in (0,1)$ to denote the linear parameterization of the geodesic going from $x$ to $y$ in $\mc{T}$.
Thus
\[
d(x, \gamma_{x,y}(t)) = t d(x,y), ~\text{and}~d(y, \gamma_{x,y}(t)) = (1-t) d(x,y).
\]
The following lemma describes a directional derivative $\omega$ of $R$ along geodesics, as well as a function $D$ that will be useful later on for analyzing convergence properties of the regularization term:
\begin{Lemma}
Define
\begin{align*}
\omega(x,y) &:= \lim_{t \to 0^+}{\frac{R(\gamma_{x,y}(t))-R(x)}{t}} \h \forall x,y \in \mathcal{T},~\text{and} \\
D(x,y) &:= R(y) - R(x) - \omega(x,y).
\end{align*}
Then
\begin{enumerate}
\item $\omega(x,y)$ is well-defined.
\item For any $x \in \mathcal{T}, |\omega(x,y)| \le 2 d(x,x^{\circ}) d(x,y)  \h \forall y \in \mathcal{T}$.
\item For all $x, y \in \mathcal{T}, D(x,y) \ge d(x,y)^2$.
\end{enumerate}
\label{lem:gradR}
\end{Lemma}

\begin{proof}
\begin{enumerate}
\item
For any $x, y \in \mathcal{T}$, we recall that the geodesic between $x$ and $y$ is composed of straight segments.
Consider the function
\[
g(t) = \frac{R(\gamma_{x,y}(t))-R(x)}{t}, \h t \in (0, 1).
\]
Let $\lambda \in (0,1)$. By convexity of $R$, we have
\begin{equation*}
R(\gamma_{x,y}(\lambda t)) = R(\gamma_{x,\gamma_{x,y}(t)}(\lambda)) \le (1-\lambda) R(x) + \lambda R(\gamma_{x,y}(t)).
\end{equation*}
This implies that
\begin{equation*}
g(\lambda t) = \frac{R(\gamma_{x,y}(\lambda t))-R(x)}{\lambda t} \le \frac{R(\gamma_{x,y}(t))-R(x)}{t} = g(t).
\end{equation*}
Hence, $g(t)$ decreases as $t \downarrow 0$. Since $R$ is locally Lipschitz, we also have that $g(t)$ is bounded.
Thus the limit of $g(t)$ as $t \downarrow 0$ exists and $\omega(x,y)$ is well-defined.

\item By \eqref{eqn:LipschitzR}, we have
\begin{multline*}
|\omega(x,y)|  \leq \lim_{t \to 0^+}\frac{|R(\gamma_{x,y}(t))-R(x)|}{t} \\
\leq \lim_{t \to 0^+}\frac{(td(x,y) + 2d(x,x^{\circ}))td(x,y)}{t} = 2d(x,x^{\circ}) d(x,y).
\end{multline*}

\item Lemma $\ref{lem:convex}$ also implies
\begin{equation*}
\frac{R(\gamma_{x,y}(t))-R(x)}{t} \le R(y) - R(x) - (1-t) d(x,y)^2.
\end{equation*}
By letting $t$ go to 0, we obtain $R(y) - R(x) - \omega(x,y) \ge d(x,y)^2$.

\end{enumerate}
\end{proof}


\section{Properties of the likelihood}
\label{sec:likelihood}
Next we investigate properties of an ``expected per-site log likelihood''
\[
\phi(\tau,q) : = \mathbb{E}_{\psi \sim P(\tau^*, q^*)}[\log P_{\tau,q}(\psi)]
\]
and construct probabilistic bounds on deviation of the empirical average per-site log likelihood from this expected quantity; these bounds are uniform over $\mc{T}_{e,g}$.
Recalling Assumption~\ref{assump:leafedges}, we have the following lemma:
\begin{Lemma}[Limit likelihood]
$(\tau^*, q^*)$ is the unique maximizer of $\phi$, and $\ell_k(\tau,q|\mb{Y}_k)/k \to \phi(\tau,q),~ \forall (\tau,q) \in \mathcal{T}$.
\label{lem:limlik}
\end{Lemma}

\begin{proof}
Note that the sequences $\mb{Y}_k$ are IID samples of $ P(\tau^*, q^*)$. By the strong law of large numbers
\begin{equation*}
\frac{1}{k}\ell_k(\tau, q|\mb{Y}_k) =  \frac{1}{k}\sum_{i=1}^k{\log P_{\tau,q}(Y_i)}  \longrightarrow  \phi(\tau,q).
\end{equation*}
for all $(\tau,q) \in \mathcal{T}$. Moreover
\begin{align*}
\phi(\tau,q) - \phi(\tau^*, q^*) &= \mathbb{E}_{\psi \sim P(\tau^*, q^*)}[\log P_{\tau,q}(\psi) - \log P_{\tau^*,q^*}(\psi)] \\
& =  - \operatorname{KL} (F(\tau^*,q^*), F(\tau,q) ) \le 0.
\end{align*}
with equality if and only if $P_{\tau,q} = P_{\tau^*,q^*}$. Hence, Assumption \ref{assump:iden} guarantees that $(\tau^*, q^*)$ is the unique maximizer of $\phi$.
\end{proof}

Next, we employ the ``covering number'' technique from machine learning \citep[see, e.g.,][]{cucker2002mathematical,cuong2013generalization} to obtain uniform bounds for the distance between the likelihood and its expected value.

\begin{Lemma}
Let $\gamma$ be the element of largest magnitude in the rate matrix $Q$. There exists a constant $C > 0$ such that for any $k \ge 3$, $\delta >0$, we have:
\begin{multline*}
\left |\frac{1}{k}\ell_k(\tau,q|\mb{Y}_k) - \phi(\tau,q)\right | \le \\
C  \left(\frac{\log k}{k}\right)^{1/2} \left( n^2\log \frac{1}{\delta} + n^3 \log(n g \gamma) + n^4 \log \frac{4}{\lambda_{e}} \right)^{1/2} \log \frac{1}{\lambda_e}
\end{multline*}
for all $(\tau,q) \in \mathcal{T}_{e,g}$ with probability greater than $1-\delta$.
Here $\lambda_e$ is a constant depending on $e$.
\label{lem:unifboundg0}
\end{Lemma}

\begin{proof}
For every assignment $\psi$ of states to the tree taxa, we have
\[
P_{\tau,q}(\psi) = \sum_{a(\psi) }{\left ( \prod_{(u,v)\in E(\tau, q)}{P^{uv}_{a_u a_v}( q_{uv})} \right ) \pi(a_{u_0})}
\]
where $a$ ranges over all extensions of $\psi$ to the internal nodes of the tree. Note that
\begin{equation*}
\left| \frac{\partial P(\psi)}{ \partial q_i}(q) \right| \le \sum_{a(\psi) }{ \left ( \prod_{(u,v)\in E(\tau, q) \setminus i}{P^{uv}_{a_ua_v}( q_{uv})} ~\left|\sum_{s}{Q_{a_u s} P_{s a_v}(q_i)}\right| \right ) \pi(a_{u_0})} \le \gamma 4^{n}.
\end{equation*}
where the bound is obtained from the fact that all probability terms are at most $1$, and that the total number of assignments $a$ to the $n-1$ inner nodes is $4^{n-1}$.

On the other hand, since the lengths of all pendant edges are bounded from below by $e$ (Assumption \ref{assump:leafedges}), $P_{\tau,q}(\psi)$ is bounded from below by the probability of the assignment when we set the value at the internal nodes to be a constant $i_0$, thus

\begin{equation*}
P_{\tau,q}(\psi) \ge \left (\min_{t \geq 0}{P_{i_0 i_0}(t)}\right)^{n-3} \left (\min_i \min_{t \ge e}{P_{i_0 i}(t)}\right)^n \pi(i_0).
\end{equation*}

We define
\[
\lambda_{e} = \left (\min_{t \geq 0}{P_{i_0 i_0}(t)}\right)^{(n-3)/n} \left (\min_i \min_{t \ge e}{P_{i_0 i}(t)}\right) \pi(i_0)^{1/n}
\]

By the Mean Value Theorem, we have
\begin{equation*}
|\log P_{\tau,q}(\psi) - \log P_{\tau,q'}(\psi)| \le c_1 \|q-q'\| \h \forall \tau,q,q', \psi
\end{equation*}
where $c_1 := \gamma 4^n/\lambda_{e}^n$, and $\|.\|$ is the Euclidean distance in $\mbb{R}^{2n-3}$.
This implies
\begin{equation}
\left |\frac{1}{k}\ell_k(\tau,q|\mb{Y}_k) - \frac{1}{k}\ell_k(\tau,q'|\mb{Y}_k)  \right|\le c_1 \|q-q'\|
\label{eqn:union1}
\end{equation}
and
\begin{equation}
|\phi(\tau,q) - \phi(\tau,q')| \le c_1 \|q-q'\|
\label{eqn:union2}
\end{equation}
for all $\tau,q ,q'$.

For each $(\tau, q) \in \mathcal{T}$, define the events
\begin{equation*}
A({\tau,q, k, \epsilon}) = \left\{\left |\frac{1}{k}\ell_k(\tau,q|\mb{Y}_k) - \phi(\tau,q)\right | > \epsilon  \right \}
\end{equation*}
and
\begin{multline*}
B({\tau,q, k, \epsilon}) = \\
 \left\{ \exists q' ~\text{such that}~ \|q'-q\| \le \frac{\epsilon}{2c_1}~~\text{and}~~ \left |\frac{1}{k}\ell_k(\tau,q'|\mb{Y}_k) - \phi(\tau,q')\right | > 2 \epsilon  \right  \}
\end{multline*}
then we have $B({\tau,q, k, \epsilon}) \subset A({\tau,q, k, \epsilon})$ by the triangle inequality, \eqref{eqn:union1}, and \eqref{eqn:union2}. We note that $0 \geq \log P_{\tau,q}(\psi) \geq n \log \lambda_e,~ \forall \tau, q, \psi.$
By Hoeffding's inequality \citep{hoeffding1963probability}, we have
\begin{equation}
\mathbb{P}[A({\tau,q, k, \epsilon})] \le 2 \exp \left (\frac{-2\epsilon^2 k}{n^2 (\log \lambda_e)^2} \right ).
\end{equation}
Note that the total number of balls of radius $\epsilon/2c_1$ required to cover $[0,g]^{2n-3}$ is bounded above by
\begin{equation}
\left(\frac{2g c_1}{\epsilon}\right)^{2n-3} (2n-3)!! = \left(\frac{2g c_1}{\epsilon}\right)^{2n-3} e^{O(n \log n)}.
\end{equation}
We deduce that for some $c_2 > 0,$
\begin{multline*}
\mathbb{P}\left[ \exists (\tau,q) \in \mathcal{T}_{e,g}: \left |\frac{1}{k}\ell_k(\tau,q|\mb{Y}_k) - \phi(\tau,q)\right | > 2 \epsilon \right] \\ 
\le 2 \exp \left (\frac{-2\epsilon^2 k}{n^2 (\log \lambda_e)^2} \right )
 \left(\frac{2g c_1}{\epsilon}\right)^{2n-3} e^{c_2 n \log n}.
 \end{multline*}
To make the left hand side less than $\delta$, we need
\[
\log(2) + (2n-3) \log(2g c_1) + c_2 n \log n + \log \frac{1}{\delta} \le \frac{2\epsilon^2 k}{n^2 (\log \lambda_e)^2} - (2n-3) \log\frac{1}{\epsilon}.
\]
We will choose $\epsilon$ such that
\[
\frac{\epsilon^2 k}{n^2 (\log \lambda_e)^2} \geq  \log(2) + (2n-3) \log(2g c _1) + c_2 n \log n + \log \frac{1}{\delta}
\]
and $\epsilon^2 k/[n^2 (\log \lambda_e)^2] \geq (2n-3) \log (1/\epsilon)$,
which is valid if we choose
\[
\epsilon = \sqrt{2 + c_2}  \sqrt{\frac{\log k}{k}} \left( n^2\log \frac{1}{\delta} + n^3 \log(n g \gamma) + n^4 \log \frac{4}{\lambda_{e}} \right)^{1/2} \log \frac{1}{\lambda_e}.
\]
\end{proof}



\begin{Lemma}
For any $0 < \delta < 1$, denote
\[
C_{n,\delta,\gamma,e,g} = C \left( n^2\log \frac{\pi^2}{6 \delta} + n^3 \log(n g \gamma) + n^4 \log \frac{4}{\lambda_{e}} \right)^{1/2} \log \frac{1}{\lambda_e}
\]
where $C$ is the constant in Lemma \ref{lem:unifboundg0}. Then
\[
\left |\frac{1}{k}\ell_k(\tau,q|\mb{Y}_k) - \phi(\tau,q)\right | \le \frac{C_{n,\delta,\gamma,e,g} \log k}{\sqrt{k}} \qquad \forall (\tau,q) \in \mathcal{T}_{e,g}, k \ge 3
\]
with probability greater than $1-\delta$.
\label{lem:ultiunifbound}
\end{Lemma}

\begin{proof}
Denote
\[
A_k = \left \{  \left |\frac{1}{k}\ell_k(\tau,q|\mb{Y}_k) - \phi(\tau,q)\right | \le C_{n,\delta/k^2,\gamma,e,g} \sqrt{\frac{ \log k}{k}},~~ \forall (\tau,q) \in \mathcal{T}_{e,g}
 \right \} .
\]
By Lemma \ref{lem:unifboundg0}, $\mbb{P}(A_k^c) \leq 6 \delta/(\pi^2 k^2)$
where $A_k^c$ is the complement set of $A_k$.
Therefore,
\begin{equation}
\mbb{P} \left ( \bigcap_{k=3}^\infty A_k \right ) = 1 - \mbb{P} \left ( \bigcup_{k=3}^\infty A_k^c \right ) \geq 1 - \sum_{k=3}^\infty{\mbb{P}(A_k^c)} \geq 1 - \sum_{k=3}^\infty {\frac{6 \delta}{\pi^2 k^2}} = 1 - \delta.
\end{equation}
\end{proof}


\section{Asymptotic theory of the regularized estimator}
\label{sec:asymptotic}
\subsection{Consistency}
\begin{Theorem}[Consistency]
Assume that the sequences  $\mb{Y}_k$ are generated from a tree  $(\tau^*, q^*) \in \mathcal{T}_{e}$ and that
\[
\alpha_k \to 0 \h \text{and} \h \alpha_k = \Omega \left ( \frac{\log k}{\sqrt{k}} \right ).
\]
Then $( \hat \tau_k, \hat q_k)$ converges to $(\tau^*, q^*)$ almost surely.
\label{conv}
\end{Theorem}

\begin{proof}
From the definition of $( \hat \tau_k, \hat q_k)$, we have
\begin{equation*}
-\frac{1}{k}\ell(\hat \tau_k, \hat q_k|\mb{Y}_k) + \alpha_k R(\hat \tau_k,\hat q_k) \le- \frac{1}{k} \ell(\tau^*, q^*|\mb{Y}_k) + \alpha_k R(\tau^*,q^*),
\end{equation*}
which implies that
\begin{multline*}
\alpha_k R(\hat \tau_k,\hat q_k)  \le  -\frac{1}{k} \ell(\tau^*,q^*|\mb{Y}_k) + \phi(\tau^*, q^*) - \phi(\tau^*, q^*) \\
+ \phi(\hat \tau_k,\hat q_k) - \phi(\hat \tau_k,\hat q_k) + \frac{1}{k} \ell(\hat \tau_k,\hat q_k|\mb{Y}_k) +   \alpha_k R(\tau^*,q^*).
\end{multline*}
By Lemma \ref{lem:ultiunifbound}, with probability greater than $1-\delta$,
\begin{equation}
\begin{aligned}
\alpha_k R(\hat \tau_k,\hat q_k) &\le - \phi(\tau^*, q^*)+ \phi(\hat \tau_k,\hat q_k)   + \frac{2C_{n,\delta,\gamma,e,g} \log k}{\sqrt{k} } +   \alpha_k R(\tau^*,q^*) \\
&\le \frac{2C_{n,\delta,\gamma,e,g} \log k}{\sqrt{k} } +   \alpha_k R(\tau^*,q^*),~~\forall k \in \mbb{N}
\label{eqn:boundR}
\end{aligned}
\end{equation}
since $(\tau^*, q^*)$ is the maximizer of $\phi$ (Lemma \ref{lem:limlik}).

Therefore, again with probability greater than $1-\delta$,
\begin{equation}
0 \leq R(\hat \tau_k,\hat q_k) \le  \frac{2C_{n,\delta,\gamma,e,g} \log k}{\alpha_k \sqrt{k} } +   R(\tau^*,q^*)\text,~~\forall k \in \mbb{N}.
\label{eqn:boundRk}
\end{equation}
We also deduce from the first inequality of Equation $\eqref{eqn:boundR}$ that
\[
\lim_{k \to \infty}{|\phi(\hat \tau_k, \hat q_k)-\phi(\tau^*, q^*)|} = 0.
\]
By the assumption on $\alpha_k$, the right hand side of \eqref{eqn:boundRk} is bounded above.
Thus the sequence $\{(\hat \tau_k,\hat q_k)\}$ lies inside a compact set, and there is a subsequence $\{(\tau_{k_m}, q_{k_m})\}$ which converges to some $(\tau',q') \in \mathcal{T}$.
The continuity of the likelihood function implies that $\phi(\tau', q') = \phi(\tau^*, q^*)$.
By Lemma \ref{lem:limlik}, we deduce that $(\tau',q') = ~(\tau^*,q^*) $.
We can repeat this argument for every subsequence of $\{(\hat \tau_k,\hat q_k)\}$ and deduce that
\begin{equation}
\lim_{k \to \infty}{(\hat \tau_k, \hat q_k)}= (\tau^*, q^*).
\label{eqn:consistency}
\end{equation}

Because (\ref{eqn:consistency}) holds with probability larger than $1 - \delta$ and $\delta$ can be arbitrarily small, we deduce that (\ref{eqn:consistency}) holds almost surely.
\end{proof}

Equation (\ref{eqn:boundR}) gives us the following corollary:
\begin{Corollary}  There exists $C_{n,\delta,\gamma,e,g} > 0$ such that with probability larger than $1 - \delta$,
\[
\phi(\tau^*,q^*) - \phi(\hat \tau_k, \hat q_k) + \alpha_k R(\hat \tau_k, \hat q_k) \le \frac{C_{n,\delta,\gamma,e,g} \log k}{\sqrt{k}} + \alpha_k R(\tau^*, q^*),~~\forall k \in \mbb{N}.
\]
\label{as2}
\end{Corollary}


\subsection{Convergence rate}

While regularized estimators are consistent in general, it is well-known that without any a priori information, their convergence can be arbitrarily slow \citep{hofmann2010interplay, kazimierski2010aspects}.
Uniform error bounds necessarily require further regularity assumptions on the (asymptotic) minimizing solution $x^*=(\tau^*,q^*)$, i.e., the data source \citep{engl1996regularization}.
Such conditions are referred to as source conditions and have played a central role in analyses of convergence of regularized estimators in various settings.
In recent publications, starting with \cite{hofmann2007convergence}, variational inequalities have become an increasingly popular way to formulate source conditions, especially in the case of non-smooth operators \citep{hohage2015verification}.
In this section, we will follow the approach of \cite{hofmann2007convergence} to derive such a variational inequality, and from that, to obtain a global error estimate for our non-smooth regularization problem.

Denote $x=(\tau,q)$ and $x^*=(\tau^*,q^*)$.
For all evolutionary models, $\phi(x)$ is an analytic function on each orthant of $\mathcal T_{e,g}$.
By the {\L}ojasiewicz inequality \citep[][Theorem 1]{ji1992global}, under any identifiable evolutionary model there exists an integer $m \geq 2$, a neighborhood $U$ around $x^*$, and some constant $C_n > 0$ such that
\begin{equation}
\phi(x^*) - \phi( x ) \ge C_n \, d(x^*,x)^m  \h \forall x \in U.
\label{eqn:regularity}
\end{equation}
We will also bound $d(x^*,x)$ in an intermediate set, defined by
\begin{equation}
U' = \{ x:  d(x^*,x)\le 4 d(x^*,x^{\circ})\} \setminus U,
\label{def:Uprime}
\end{equation}
using $\eta_{x^*} = \inf_{x\in U'}{\left(\phi(x^*) -  \phi( x)\right)}  =\phi(x^*) - \sup_{x\in U'}{\phi( x)} > 0$.
We will use these to bound the directional derivative $\omega$ of $R$ defined in Lemma~\ref{lem:gradR}:

\begin{Lemma}[Variational source condition]
There exists $\beta > 0$ such that
\[
- \omega(x^*,x)  \le  \frac{1}{2} D(x^*, x) + \beta (\phi(x^*) - \phi(x))^{1/m}  \h \forall x \in \mathcal{T}.
\]
\label{lem:boundgradR}
\end{Lemma}
\begin{proof}
The proof comes from controlling $d(x^*, x)$ in three regimes: in $U$ using \eqref{eqn:regularity}, when $d(x^*, x)$ is large with Lemma $\ref{lem:gradR}$, and in the intermediate regime $U'$.
If $U'$ is non-empty, we have
\begin{equation}
\phi(x^*) - \phi( x ) \ge \frac{\eta_{x^*}}{[4 d(x^*,x^{\circ})]^m} d(x^*,x)^m, \h \forall x \in U'.
\label{source1}
\end{equation}
Moreover, when $d(x,x^*) > 4 d(x^*,x^{\circ})$, we have
\begin{equation}
d(x,x^*) < \frac{1}{4 d(x^*,x^{\circ})} d(x^*,x)^2 .
\label{source2}
\end{equation}
Combining Equations $\eqref{eqn:regularity}$, $\eqref{source1}$ and $\eqref{source2}$ and using Lemma \ref{lem:gradR},  we deduce that $- \omega(x^*,x) \le   2 d(x^*,x^{\circ}) d(x^*, x) \le D(x^*, x)/2 + \beta (\phi(x^*) - \phi(x))^{1/m}$ where $\beta = 2 d(x^*,x^{\circ}) \max \left \{ \frac{1}{C_n^{1/m}}, \frac{4 d(x^*,x^{\circ}) }{\eta_{x^*}^{1/m}} \right \}$.
\end{proof}

This variational inequality enables us to derive the following error estimates.

\begin{Theorem}[Global error of regularized estimator]
Let $x_k = (\hat \tau_k, \hat q_k)$ be the minimizer of $\eqref{eq1}$ on $\mc{T}_{e,g}$.
Denote
\[
C_{n,\delta,\gamma,e,g, m, x^*, x^{\circ}} = \sqrt{2} \left(C_{n,\delta,\gamma,e,g}  + \frac{(m-1)\beta^{m/(m-1)}}{m^{m/(m-1)}}\right)^{1/2},
\]
where $\beta = 2 d(x^*,x^{\circ}) \max \left \{ \frac{1}{C_n^{1/m}}, \frac{4 d(x^*,x^{\circ}) }{\eta_{x^*}^{1/m}} \right \}$ and $C_{n,\delta,\gamma,e,g}$ is the constant in Lemma $\ref{lem:ultiunifbound}$.
For any $\delta>0$, we have
\[
d(x^*, x_k) \le C_{n,\delta,\gamma,e,g, m, x^*, x^{\circ}}\left(   \frac{\log k}{\alpha_k \sqrt{k} } +   \alpha_k^{1/(m-1)} \right)^{1/2}
\]
with probability greater than $1-\delta$.
\label{rate}
\end{Theorem}

\begin{proof}
From Corollary \ref{as2}, we note that with probability greater than $1-\delta$, we have
\begin{equation*}
\phi(x^*) - \phi(x_k)  + \alpha_k R(x_k) \le \frac{C_{n,\delta,\gamma,e,g} \log k}{\sqrt{k} } +   \alpha_k R(x^*), ~~\forall k \in \mbb{N}
\end{equation*}
which implies
\begin{multline}
\phi(x^*) - \phi(x_k) + \alpha_k D(x^*, x_k)  \le  \frac{C_{n,\delta,\gamma,e,g} \log k}{\sqrt{k} } \\
+   \alpha_k \left( R(x^*) - R(x_k)+D(x^*, x_k) \right), ~~\forall k \in \mbb{N}.
\label{eqn:eqn1}
\end{multline}
By Lemma \ref{lem:boundgradR}, we have
\begin{multline}
R(x^*) - R(x_k) + D(x^*, x_k)  \\
= - \omega (x^*,x) \le  \frac{1}{2} D(x^*, x_k) + \beta (\phi(x^*) - \phi(x_k))^{1/m}.
\label{eqn:eqn2}
\end{multline}

Combining (\ref{eqn:eqn1}) and (\ref{eqn:eqn2}), we obtain
\begin{equation*}
\frac{\alpha_k}{2} D(x^*, x_k) \le \frac{C_{n,\delta,\gamma,e,g} \log k}{\sqrt{k}} + \alpha_k \beta (\phi(x^*) - \phi(x_k))^{1/m} - (\phi(x^*) - \phi(x_k)).
\end{equation*}

Denote $\overline{m} = m/(m-1)$. Applying Young's inequality, we get
\begin{equation*}
\frac{1}{m}\left((\phi(x^*) - \phi(x_k))^{1/m} \right)^m + \frac{1}{\overline{m}} \left( \frac{\alpha_k \beta}{m}\right)^{\overline{m}} \ge \frac{\alpha_k \beta}{m} (\phi(x^*) - \phi(x_k))^{1/m}.
\end{equation*}

We deduce that
\begin{equation*}
D(x^*, x_k) \le  2 \left(  \frac{C_{n,\delta,\gamma,e,g} \log k}{\alpha_k \sqrt{k} } + \frac{(m-1)}{\alpha_k} \left( \frac{\alpha_k \beta}{m}\right)^{\overline{m}}\right)
\end{equation*}
and the proof is complete by another application of Lemma~\ref{lem:gradR}.
\end{proof}

By choosing $\alpha_k = k^{(1-m)/(2m)}$ and applying Theorem~\ref{rate}, we have the following corollary:
\begin{Corollary}
For any $\delta>0$, we have $d(x^*, x_k)  = \mc{O} \left ( (\log k)/(\sqrt[4m]{k}) \right )$ with probability greater than $1-\delta$.
\end{Corollary}

Next we describe a sufficient condition for $m$ achieving its minimum of 2, which we then show holds everywhere for some simple phylogenetic models.
Define the expected Fisher information matrix $I_\tau(q)$ by
\[
I_\tau(q)_{ij} = -\mathbb{E}_{\psi \sim P(\tau^*, q^*)}\left[ \frac{\partial^2}{\partial q_i  \partial q_j} \log P_{\tau,q}(\psi) \right].
\]
Since $\phi$ is analytic around $(\tau^*, q^*)$ and $\nabla \phi(\tau^*,q^*)=0$, we deduce that around $(\tau^*,q^*)$,
\[
\phi(\tau^*, q) - \phi(\tau^*, q^*) = - (q-q^*)^t I_{\tau^*}(q^*) (q-q^*) +  o\left( \|q^*-q\|^2  \right).
\]
This gives a sufficient condition for $m=2$:

\begin{Lemma}
Assume that the Fisher information matrix $I_{\tau^*}(q)$ is positive-definite at $q^*$. Then there exists a neighborhood $U$ around $(\tau^*, q^*)$ and some constant $C_n > 0$ such that $\phi(\tau^*, q^*) - \phi(\tau^*, q) \ge C_n \|q^*-q\|^2  ~\forall (\tau^*,q) \in U.$
\label{lem:kIsTwo}
\end{Lemma}

The condition that the Fisher information matrix is positive-definite at $q^*$ is not restrictive. Indeed, the matrix $I = I_{\tau^*}(q^*)$ can be rewritten in the form
\[
I_{ij} = \mathbb{E}\left[ \left(\frac{\partial}{\partial q_i } \log P_{\tau^*,q^*}(\psi) \right) \left(\frac{\partial}{\partial q_j } \log P_{\tau^*,q^*}(\psi) \right)\right].
\]
For every nonzero vector $u \in \mathcal{R}^n$, we have
\begin{align}
\sum_{ij}{u_i I_{ij} u_j} &= \sum_{ij}{u_i \mathbb{E}\left[ \left(\frac{\partial}{\partial q_i } \log P_{\tau^*,q^*}(\psi) \right) \left(\frac{\partial}{\partial q_j } \log P_{\tau^*,q^*}(\psi) \right)\right]u_j} \nonumber \\
&=  \mathbb{E}\left[ \left(\sum_{i} u_i \frac{\partial}{\partial q_i } \log P_{\tau^*,q^*}(\psi) \right) \left(\sum_{j} u_j \frac{\partial}{\partial q_j } \log P_{\tau^*,q^*}(\psi) \right)\right] \nonumber \\
&= \mathbb{E}\left[ \left(\sum_{i} u_i \frac{\partial}{\partial q_i } \log P_{\tau^*,q^*}(\psi) \right)^2 \right] \ge 0
\end{align}
with equality when
\begin{equation}
\sum_{i} u_i \frac{\partial}{\partial q_i } \log P_{\tau^*,q^*}(\psi) = 0 \h \forall \psi.
\label{eqn:deficient}
\end{equation}
If we consider the mapping
\begin{equation}
\begin{aligned}
G: \Delta &\rightarrow \mathbb{R}^{4^n} \\
[G(\tau,q)]_{\psi} &=  \log P_{\tau,q}(\psi).
\label{eqg}
\end{aligned}
\end{equation}
on some compact neighborhood $\Delta \subset \mc{T}$ of $(\tau^*,q^*)$ that is contained in the same orthant as $(\tau^*, q^*)$, then \eqref{eqn:deficient} implies that each column of the Jacobian of the map $G$ at $q^*$ is orthogonal to $u$.
Thus, in this case the Jacobian $J_G(q^*)$ is rank-deficient.
Hence, in order to verify the criteria of Lemma~\ref{lem:kIsTwo}, we just need to prove that $J_G(q^*)$ has full rank.
This condition can be verified for various classes of model including the $r$-state symmetric models \citep{semple2003phylogenetics} and the popular Felsenstein 1984 (F84) model first implemented in the DNAML software program \citep{felsenstein1984dnaml}.

\begin{Lemma}[Information-regularity]
The Fisher information matrix $I_\tau(q)$ is everywhere positive-definite for $r$-state symmetric models (e.g. Jukes-Cantor), and F84.
\label{inform-reg}
\end{Lemma}

\begin{proof}
\citet{wang2004maximum} proves that $I_\tau(q)$ is positive definite for F84. For $r$-state symmetric models, we note that $G$ (as defined in $\eqref{eqg}$) is continuous and injective. This shows that $G^{-1}$ is well-defined and continuous. If $G^{-1}$ is smooth, then the Jacobian of $G$ at $q^*$ has full rank. This can be shown by the following commutative diagram
\[
\begin{array}[c]{ccc}
\mc{A}&\stackrel{Hadamard}{\leftrightarrow}&\mc{B}\\
\updownarrow\scriptstyle{r_1}&&\updownarrow\scriptstyle{r_2}\\
\mathcal{T}&\stackrel{G}{\rightarrow}&\text{Im}[G(.)]
\end{array}
\]
where $r_1$ maps $(\tau, q)$ to its edge length spectrum $Q(\tau,q) \in \mc{A}$, $r_2$ maps $G(\tau,q)$ to its corresponding sequence spectrum $\Phi \in \mc{B}$, and the Hadamard-exponential conjugation \citep{semple2003phylogenetics, felsenstein2004inferring} provides a smooth conversion between the two spectra. (A detailed description of the Hadamard-exponential conjugation will be provided in the next section.)
We deduce that $G^{-1}$ is smooth and the result holds.
\end{proof}


\section{Special case: Jukes-Cantor model}
\label{sec:complexity}

In this section, we aim to quantify the two constants $C_n$ and $\eta_{x^*}$ defined before Lemma~\ref{lem:boundgradR} to derive an explicit bound on the rate of convergence of the regularized estimator under the Jukes-Cantor model using Theorem~\ref{rate}.
From that, we will obtain a sequence length bound to recover the true tree topology.
While we provide details for the Jukes-Cantor model, our results can be generalized to $r$-state symmetric models.
We will take an information-theoretic perspective, working in terms of KL divergences rather than the equivalent formulation using differences in the expected per-site log likelihood function $\phi$.

Without loss of generality, we assume that the substitution rate $\mu = 1$.
We recall that by Lemma $\ref{inform-reg}$, Equation $\eqref{eqn:regularity}$ holds with $m=2$.


\subsection{Convergence rate of the regularized estimator under the Jukes-Cantor model}

Recall that the constants $C_n$ and $\eta_{x^*}$ in Theorem \ref{rate} are defined in such a way that
$ \operatorname{KL}(P_{x*}, P_{x})  = \phi(x^*) - \phi( x ) \ge C_n \, d(x^*,x)^m  ~~ \forall x \in U$, and
\[
\eta_{x^*} =  \inf_{x\in U'}{\operatorname{KL}(P_{x*}, P_{x})} = \inf_{x\in U'}{\left(\phi(x^*) -  \phi( x)\right)}  =\phi(x^*) - \sup_{x\in U'}{\phi( x)},
\]
for some neighborhood $U$ around $x^*$ and $U'$ as defined in \eqref{def:Uprime}.

In order to quantify these constants, we need to derive lower bounds of $\operatorname{KL}$ divergence between the site expected site pattern frequencies of trees.
The underlying ideas behind the following somewhat technical proofs can be sketched as follows.
We define $U = \{ x: d(x,x^*) < f \}$ where $f$ denotes the length of the shortest edge of $x^*$, and consider two different scenarios:

\begin{itemize}
\item[(i)] If $x \in U$, then $x$ has the same topology as $x^*$.
We then prove (in Lemma~\ref{lem:KL4taxon}) that we can find four taxa such that the restriction of $x^*$ and $x$ to these four taxa induces two quartets $s^*$ and $s$ of the same topology such that
 \[
 \frac{1}{2n-3}d(x^*, x) \le d(s^*, s) \le d(x^*, x), ~\text{and}~ \operatorname{KL}(P_{x^*},P_x) \geq \operatorname{KL}{(P_{s^*},P_s)}
 \]
by which a bound for the 4-taxon case gives a bound for the general case.
This case provides a lower bound for $C_n$.
\item[(ii)] If $x \in \mc{T}_{e,g} \setminus U$, we can construct two quartets $s^*$ and $s$ such that
\[
\operatorname{diam}(s^*) \leq 4g \log_2(n),~~ \text{and}~~d(s,s^*) \geq f/(2n-3)
\]
which also reduces our bound to the 4-taxon case.
The resulting bound on $\eta_{x^*}$ need not depend on the distance $d(x^*,x)$.
Details are provided in Lemma~$\ref{lem:KLoutside}$.
\end{itemize}

In both cases, a lower bound of $\operatorname{KL}{(P_{s^*},P_s)}$ by $d(s,s^*)$ is needed.
This bound is partially derived in Lemma $\ref{lem:quartet}$ under the assumption that the diameters of $s^*$ and $s$ are bounded from above.
While this assumption works fine for case (i), an upper bound for $\operatorname{diam}(s)$ can not be derived for case (ii), which prompts us to consider two different sub-cases for the value of $\operatorname{diam}(s)$ in the proof of Lemma $\ref{lem:KLoutside}$.

\begin{figure}
\centering
  \includegraphics[width=0.5\linewidth]{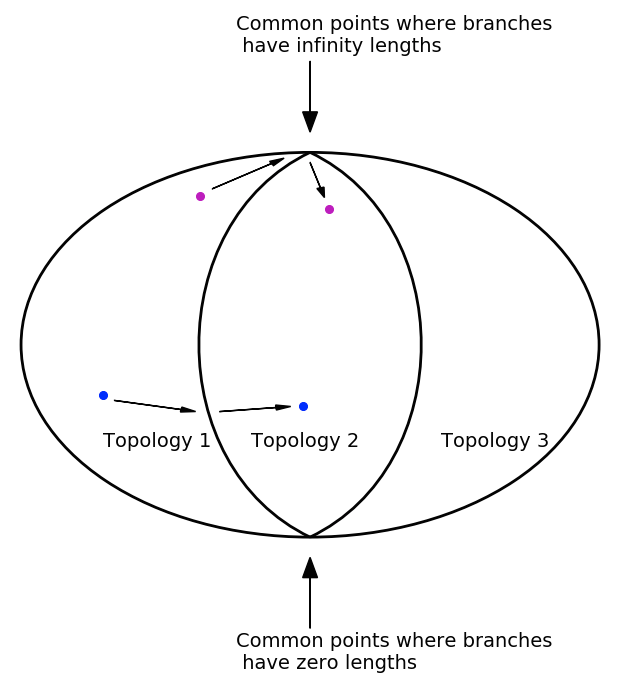}
  \caption{Visualization of the phylogenetic orange:  without the bound on the diameter of the trees, one can easily construct two trees $x_1$ and $x_2$ such that $P_{x_1}$ and $P_{x_2}$ are arbitrarily close and $d_{BHV}(x_1, x_2)$ is arbitrarily large.}
  \label{fig:orange}
\end{figure}

Before proceeding to provide the detailed proofs, we note that a lower bound of $\operatorname{KL}{(P_{s^*},P_s)}$ by $d(s,s^*)$ can not be obtained without the assumption that the diameters of the quartets are bounded.
This can be seen by considering the phylogenetic ``orange'' \citep{kim2000slicing,Moulton2004-rd}: the space of all leaf-node distributions $P_{\tau, q}$ for $(\tau, q) \in \mathcal{T}$.
As sets of branch lengths become large, the corresponding point on the orange moves towards the top point of this compact set, which is that induced by independent samples from the stationary distribution on states.
For that reason, without the bound on the diameter of the trees, one can easily construct two trees $x_1$ and $x_2$ such that $P_{x_1}$ and $P_{x_2}$ are arbitrarily close and $d(x_1, x_2)$ is arbitrarily large, rendering the bound impossible.

\subsubsection{Lemmas on quartets}
However, when the diameter of quartets is bounded above, we can use Hadamard conjugation \citep{Hendy1993-gt, Steel1998-in, semple2003phylogenetics, felsenstein2004inferring} to bound KL divergence below in terms of the BHV distance as follows.
\begin{Lemma}
Let $\mathcal{K}_b$ be the BHV space of all quartets with diameter bounded from above by $b$.
We have
\[
\operatorname{KL}(P_{\tau_1,q_1}, P_{\tau_2,q_2}) \ge 256 e^{-8b} d((\tau_1,q_1),(\tau_2,q_2))^2
\]
for all $(\tau_1,q_1), (\tau_2,q_2) \in \mathcal{K}_b$.
\label{lem:quartet}
\end{Lemma}

\begin{proof}
First we recall that the \emph{branch length spectrum} $\mb{a}$ of a quartet $\tau$ is defined as
\[ \mb{a}_{A}=  \begin{cases}
     q_e& \textrm{ if $e \in E(\tau)$ induces the split $A|(\{1,2,3,4\} \setminus A)$} \\
      -\sum_{e\in E(\tau)}{q_e} & \textrm{ if $A = \emptyset$} \\
       0& \textrm{otherwise} \\
   \end{cases} \]
for any subset $A \subset\{1, 2, 3\}$.
We will use $r_1$ to denote the invertible map from a set of branch lengths to the branch length spectrum.
The \emph{sequence spectrum} is derived from the vector of single-site patterns at the leaf nodes by summing the probabilities for site patterns that induce identical splits.
The mapping from an expected site pattern frequency (the image of $G$) to its sequence spectrum is an invertible transformation we will call $r_2$ \citep[see \S 2.2.2,][]{Steel1998-in}.

Hadamard exponential conjugation induces the commutative diagram mentioned previously:
\[\begin{array}[c]{ccc}
\mc{A}&\stackrel{Hadamard}{\leftrightarrow}&\mc{B} \\
\updownarrow\scriptstyle{r_1}&&\updownarrow\scriptstyle{r_2}\\
\mathcal{K}_b&\stackrel{G}{\rightarrow}&\text{Im}[G(\cdot)]
\end{array}\]
where $r_1$ maps $(\tau, q)$ to its edge length spectrum $Q(\tau,q) \in \mc{A} = \text{Im}[r_1(\cdot)]$, $r_2$ maps $G(\tau,q)$ to its corresponding sequence spectrum $\Phi \in \mc{B}=\text{Im}[r_2(\cdot)]$, and the Hadamard-exponential conjugation provides a smooth conversion between the two spectra by the formula
\begin{equation}
\mb{a} = \mb{H}^{-1} \log (\mb{H} \mb{b}), ~\text{and}~~ \mb{b} = \mb{H}^{-1} \exp (\mb{H} \mb{a})
\end{equation}
where $\mb{a} \in \mc{A}$, $\mb{b} \in \mc{B}$ are vectors of dimension $8$, $\log$ and $\exp$  are applied component-wise, and $\mb{H}$ is a Hadamard matrix of rank $8$ defined as follows:
\[
\mb{H}_1 = [ 1 ],~~ \mb{H}_{i+1} = 
\begin{bmatrix}
    \mb{H}_i       &  \mb{H}_i \\
     \mb{H}_i       &  - \mb{H}_i 
\end{bmatrix}, ~~
\text{and}~~ \mb{H} = \mb{H}_4.
\]

Hence, if $(\tau_1,q_1), (\tau_2,q_2) \in \mathcal{K}_b$, then the entries of $\mb{H} \mb{a}$ satisfy
\begin{equation*}
[\mb{H}\mb{a}]_i = \sum_{j}{H_{ij} a_j} \ge \sum_{j}{- |a_j|} = - 2\sum_{e \in E(\tau)}{q_e} \ge - 4 \operatorname{diam}(\tau, q) \ge -4 b \h \forall i.
\end{equation*}
The first inequality holds because $|H_{ij}| = 1$ for all $i$ and $j$. 

Note that
\begin{equation}
\begin{aligned}
&\| \mb{H} x \|_1 \leq 8 \|x\|_1, ~\text{and}~  \| \mb{H}^{-1} x \|_1 = \left \| \frac{1}{8} \mb{H} x \right \|_1 \leq \|x\|_1, ~ \forall x \in \mbb{R}^8 \\
&[\mb{H}\mb{b}]_i = \exp([\mb{H}\mb{a}]_i) \geq e^{-4 b}~~ \forall i,
\end{aligned}
\end{equation}
where $\|\cdot\|_1$ is the standard $L_1$ norm.
From \citet{amenta2007approximating}, we have
\[
d((\tau_1,q_1),(\tau_2,q_2)) \le \sqrt{2} ~d_{\operatorname{BS}}((\tau_1,q_1),(\tau_2,q_2)) \le \sqrt{2} \|\mb{a}_1 - \mb{a}_2\|_1
\]
where $ \|\mb{a}_1 - \mb{a}_2\|_1$ denotes the $L_1$ distance between 
$\mb{a}_1$ and $\mb{a}_2$.
Thus
\begin{align}
d((\tau_1,q_1),(\tau_2,q_2)) &= \sqrt{2}\|\mb{a}_1 - \mb{a}_2\|_1 =\sqrt{2}  \| \mb{H}^{-1} (\log(\mb{H} \mb{b}_1)) - \mb{H}^{-1} (\log(\mb{H} \mb{b}_2)) \|_1 \nonumber \\
& \leq \sqrt{2} \|\log(\mb{H} \mb{b}_1) - \log(\mb{H} \mb{b}_2)\|_1 \le \sqrt{2} e^{4b} \|\mb{H} \mb{b}_1 - \mb{H} \mb{b}_2\|_1 \nonumber \\
&\le 8 \sqrt{2} e^{4b} \|\mb{b}_1 -  \mb{b}_2 \|_1 = 16 \sqrt{2} e^{4b} d_{\operatorname{TV}} (P_{\tau_1, q_1}, P_{\tau_2, q_2}) \nonumber \\
&\leq16 e^{4b} \sqrt{\operatorname{KL}(P_{\tau_1,q_1}, P_{\tau_2,q_2})}
\end{align}
where $d_{\operatorname{TV}}$ is the total variation distance and the last line is by Pinsker's inequality.
\end{proof}

\begin{Lemma}
Let $U = \{ x: d(x,x^*) < f \}$ where $f$ denotes the length of the shortest edge of $x^*$.
\begin{itemize}
\item[(i)] If $x \in U$, then there exist four taxa such that the restriction of $x^*$ and $x$ to these four taxa induces two quartets $s^*$ and $s$ of the same topology such that
 \[
\operatorname{diam}(s^*) \leq 4g \log_2(n) \h \text{and} \h  \frac{1}{2n-3}d(x^*, x) \le d(s^*, s) \le d(x^*, x).
 \]
\item[(ii)] If $x \in \mc{T}_{e,g} \setminus U$, then there exist four taxa such that the restriction of $x^*$ and $x$ to these four taxa induces two quartets $s^*$ and $s$ such that
\[
\operatorname{diam}(s^*) \leq 4g \log_2(n) \h \text{and} \h  f/(2n-3) \leq d(s,s^*) .
\]
\end{itemize}
Moreover, in both cases, we have $\operatorname{KL}(P_{x^*},P_x) \geq \operatorname{KL}{(P_{s^*},P_s)}$.
\label{lem:KL4taxon}
\end{Lemma}

\begin{proof}
(i)   It is obvious that $x$ has the same topology as $x^*$ if $x \in U$. For every internal edge $i$, we define a quartet $s^*_i$ by the following procedure: delete $i$ and the edges touching it to obtain $4$ rooted subtrees, then select one leaf that is closest to the root within each subtree to form the quartet $s^*_i$.
Following the notation of \citet{erdos1999few}, we 
refer to such a quartet as a \emph{representative quartet} of the tree $x$.
Let $s_i$ be the quartet of tree $x$ formed of the same four leaves.
We call $s_i$ the \emph{associated quartet} of $s^*_i$.

Since $x$ and $x^*$ have the same topology, $s^*_i$ and $s_i$ contain the same set of edges.
Moreover, every edge of $x^*$ belongs to at least one $s^*_i$. Hence
\begin{equation*}
d(x^*,x) \leq \sum_i d(s^*_i,s_i).
\end{equation*}
Therefore, there exists a quartet $s^*$ and its associated quartet $s$ such that
\begin{equation*}
\frac{1}{2n-3} d(x^*,x) \leq d(s^*,s).
\end{equation*}

By the construction, the pendant edge lengths of the representative quartet $s^*$ are bounded from above by $g (\op{depth}(x^*) + 1)$ \citep{erdos1999few}.
Moreover, we can bound $\op{depth}(x^*)$ above by $\log_2(n)$ \citep{erdos1999few, csuros2002fast, mossel2004phase}.
Hence,
\[
\operatorname{diam}(s^*) \leq 2 (\op{depth}(x^*) + 1) g +g  \leq g(2\log_2(n) + 3) \leq 4g \log_2(n).
\]
It is worth noting that while our notation of depth of a tree is the same as those defined in \citet{erdos1999few, csuros2002fast, mossel2004phase}, our notation of diameter incorporates branch lengths.
Thus $g$ appears in the bound on $\operatorname{diam}(s^*)$.

(ii) For every $x \in \mathcal{T}(e, g) \setminus U$, let $s^*$ and $s$ be constructed as per the previous case if $x$ has the same topology as $x^*$.
Otherwise, using Lemma 2 in \citep{erdos1999few}, we deduce that there exists a set of four leaves $M$ such that the topologies of $s^*=x^*|_M$ and $s = x|_M$ are different and that $s^*$ is a representative quartet of the tree $x^*$.
Using the same arguments as in part (i), we have $\operatorname{diam}(s^*) \leq 4g \log_2(n)$.

In both scenarios, we can choose $s^*$ and $s$ such that
\[
\operatorname{diam}(s^*) \leq 4g \log_2(n),~~ \text{and}~~f/(2n-3) \leq d(s,s^*).
\]
Finally, from the ``information processing inequality'' \citep[see, e.g., Theorem 9 of][]{van2014renyi}, we have $\operatorname{KL}(P_{x^*},P_x) \geq \operatorname{KL}{(P_{s^*},P_s)}$.

\end{proof}

\subsubsection{Convergence rate}
Next we will bound $d(x^*, x_k)$ in terms of sequence length by obtaining values for the constants $C_n$ and $\eta_{x^*}$.
The following lemma states that we can take $C_n = 64 n^{-32g/\log2 - 8\sqrt{3}f - 2}$.
\begin{Lemma}
Let $U = \{ x: d(x,x^*) < f \}$ where $f$ denotes the length of the shortest edge of $x^*$, we have $\operatorname{KL}(P_{x^*},P_x) \geq 64 n^{-32g/\log2 - 8\sqrt{3}f - 2} d(x^*,x)^2$, for all $x \in U.$
\label{lem:KLinside}
\end{Lemma}

\begin{proof}
Since $d(s^*,s) \leq d(x^*,x) \leq f $, we have $\operatorname{diam}(s) \leq \operatorname{diam}(s^*) + \sqrt{3}f$.
Lemma $\ref{lem:KL4taxon}$ gives quartets $s$ and $s^*$ with diameter $\operatorname{diam}(s^*) \leq 4g \log_2(n)$ and
\[
\operatorname{diam}(s) \leq \operatorname{diam}(s^*)  + \sqrt{3}f  \leq 4g \log_2(n)+\sqrt{3}f =: b.
\]
This bound can be rearranged to take the form $e^{-8b} \geq n^{-32g / \log 2 - 8 \sqrt{3}f}$, which when combined with Lemmas~\ref{lem:quartet} and \ref{lem:KL4taxon} gives
\[
\operatorname{KL}(P_{x^*},P_x) \ge \operatorname{KL}{(P_{s^*},P_s)} \geq 64 n^{-32g/\log2 - 8\sqrt{3}f - 2} d(x^*,x)^2.
\]

\end{proof}

Next we obtain a value for $\eta_{x^*}$, which bounds $\operatorname{KL}(P_{x^*}, P_x)$ below in $U'$:
\begin{Lemma}
Let $U = \{ x: d(x,x^*) < f \}$ where $f$ denotes the length of the shortest edge of $x^*$, we have
$\operatorname{KL}(P_{x^*}, P_x) \ge C_f n^{-64g/\log2 - 2}, ~\forall n \ge 4$ and $x \in \mathcal{T}(e, g) \setminus U$, where $C_f = \min \{64 f^2, 72[1 - 4^{-16f/(3 \log 2)}]^2 \}$.
\label{lem:KLoutside}
\end{Lemma}

\begin{proof}
If $\operatorname{diam}(s) \leq \operatorname{diam}(s^*) + 4g \log_2(n)$, then both $\operatorname{diam}(s)$ and $\operatorname{diam}(s^*)$ are bounded from above by $8g \log_2(n)$.
Thus,
\begin{equation}
\operatorname{KL}(P_{x^*}, P_{x}) \ge  \operatorname{KL}{(P_{s^*},P_s)}  \ge 64 n^{-64g/\log2 - 2} f^2
\label{eqn:upperb1}
\end{equation}
by Lemma \ref{lem:quartet} and Lemma $\ref{lem:KL4taxon}$.
If not, there exist two taxa $S_1$ and $S_2$ such that $ \iota := \iota_{x}(S_1, S_2) > \operatorname{diam}(s^*) + 4g \log_2(n)$, and $\iota^* := \iota_{x^*}(S_1, S_2) \leq \operatorname{diam}(s^*)$ where $\iota_{x}(S_1, S_2)$ is the sum of the branch lengths between $S_1$ and $S_2$ on tree $x$.
We have
\begin{equation*}
\operatorname{KL}(P_{x^*}, P_{x}) \ge \operatorname{KL}(P_{x^*|_{\{S_1, S_2\}}}, P_{x|_{\{S_1, S_2\}}}) \geq 2 d_{\operatorname{TV}} (P_{x^*|_{\{S_1, S_2\}}} , P_{x|_{\{S_1, S_2\}}} )^2.
\end{equation*}
Under the Jukes-Cantor model,
\[
P_{x|_{\{S_1, S_2\}}}(uv) =
\begin{cases}
\frac{1}{4} + \frac{3}{4}e^{-4\iota/3}, & \mbox{if } u=v  \\
\frac{1}{4} - \frac{1}{4}e^{-4\iota/3}, & \mbox{if } u \ne v.
\end{cases}
\]
Therefore
\begin{multline}
d_{\operatorname{TV}} (P_{x^*|_{\{S_1, S_2\}}} , P_{x|_{\{S_1, S_2\}}} ) = 6(e^{-4\iota^*/3} - e^{-4\iota/3}) \\
\geq 6n^{-16g/(3 \log 2)} \left [ 1 - n^{-16g/(3 \log 2)} \right ] \geq 6n^{-16g/(3 \log 2)} \left [1 - 4^{-16f/(3 \log 2)} \right ].
\label{eqn:upperb2}
\end{multline}

The lemma follows directly from \eqref{eqn:upperb1} and \eqref{eqn:upperb2}.
\end{proof}

By Lemma \ref{lem:KLinside}, Lemma \ref{lem:KLoutside}, and Theorem \ref{rate}, we get the following theorem:

\begin{Theorem}
Let $x_k = (\hat \tau_k, \hat q_k)$ be the minimizer of $\eqref{eq1}$ on $\mathcal{T}_{e, g}$.
For any $\delta>0$,
 \[
d(x^*, x_k) \le \sqrt{2} \left(  C_{n,\delta,\gamma,e,g} + \frac{\beta^{2}}{4}\right)^{1/2} \left(   \frac{\log k}{\alpha_k \sqrt{k} } +   \alpha_k \right)^{1/2}
\]
with probability greater than $1-\delta$, where
$C_{n,\delta,\gamma,e,g}$ is the constant in Lemma $\ref{lem:ultiunifbound}$, $C_f$ is the constant in Lemma \ref{lem:KLoutside} and
\[
\beta = \max \left \{32 d(x^*,x^{\circ})  n^{16g/\log2 + 4\sqrt{3}f + 1}, \frac{8}{\sqrt{C_f}} d(x^*,x^{\circ})^2 n^{32g/\log2 + 1}  \right \}.
\]
\label{thm:JukesCantor}
\end{Theorem}

If we have a lower bound $f$ for all inner edges, all constants in Theorem \ref{thm:JukesCantor} can be evaluated.
Hence, we can construct a conservative confidence region for the regularized ML estimator.
This enables us to assess the support of tree splits from data.
We also note that the term $d(x^*,x^{\circ}) ^2$ appears in the constants of Theorem $\ref{thm:JukesCantor}$, illustrating the fact that informative prior knowledge (characterized by a good regularizing tree) makes inference better.
On the other hand, since the NNI-diameter of the tree space is of order $\mc{O}(n \log n)$ \citep{roch2015phase}, $d(x^*,x^{\circ})$ is upper bound by $(\mc{O}(g n\log n))$.
This implies that even a horrid regularizing tree would help stabilizing an estimator globally without greatly affect the convergence rate.


\subsection{Sample complexity for recovering the true tree topology}

We can use the explicit convergence bound to derive a sequence length requirement for recovering the true topology.
From Theorem \ref{thm:JukesCantor}, if the length of the observed sequence $k$ satisfies
\[
\sqrt{2} \left(  C_{n,\delta,\gamma,e,g} + \frac{\beta^{2}}{4}\right)^{1/2}  \left(   \frac{\log k}{\alpha_k \sqrt{k} } +   \alpha_k \right)^{1/2} \leq  f,
\]
then the topology of reconstructed tree 
using the regularized estimator is correct with probability greater than $1-\delta$. By choosing $\alpha_k = k^{-1/4}$, we have

\begin{Theorem}
For any $\delta>0$ and $0 < \nu < 1/4$, if the length of the observed sequence $k$ satisfies
\[
k \sim \left ( \left [  \frac{1}{f^2} \left(\log \frac{1}{\lambda_e} \right)^2 \left ( \log \frac{4 g \gamma}{\lambda_{e} \delta} \right ) n^4  + \frac{\beta(n, g, f)^2}{4f^2} \right ]^{\frac{4}{1 - 4 \nu}} \right )
\]
where $\beta(n,g,f)$ is defined as in Theorem \ref{thm:JukesCantor}, then the topology of the reconstructed tree using the regularized estimator is correct with probability greater than $1-\delta$.
\label{thm:complexity}
\end{Theorem}

\begin{proof}
The topology of reconstructed tree $x_k$ using the regularized estimator is correct with probability greater than $1-\delta$ if
\begin{equation}
(2/f^2) \left[  C_{n,\delta,\gamma,e,g} + (\beta^{2}/4) \right ]   \leq  k^{1/4}/( \log k +1).
\label{nu}
\end{equation}
Moreover, for all $0<\nu<1/4$, there exists $C_{\nu}$ such that $k^{1/4}/(\log k +1) \ge C_{\nu} k^{\frac{1}{4}-\nu}$.
Thus, $\eqref{nu}$ holds if
\[
k \ge \left[\frac{2}{C_{\nu}f^2}\left(  C_{n,\delta,\gamma,e,g} + \frac{\beta^{2}}{4}\right) \right]^{\frac{4}{1-4\nu}}.
\]
This completes the proof.
\end{proof}

When the model parameters (including $e$ and $g$) and the level of confidence $\delta$ are fixed, the sample complexity to recover the tree topology is thus

\[
k =  \mc{O} \left(\frac{n^{ \mc{O}(g)}}{f^{16/(1-4\nu)}}\right)
\]
for any $0 < \nu < 1/4$. Since our regularized-ML algorithm can reconstruct the topology from input sequences of length polynomial in the number of terminal taxa $n$, it falls in to the class of so-called \emph{fast converging} tree reconstruction algorithms \citep{huson1999disk, warnow2001absolute}.

\citet{gronau2012fast} point out that the main drawback of most fast converging algorithms is that they perform poorly on trees with indistinguishable/very short edges whereas a single short edge may affect the ability of the algorithm to support any other long edges of the generating tree.
Such algorithms share an ``all or nothing'' nature: when failing to reconstruct an edge, some of these algorithms do not return any tree, while others return a tree that may contain many faulty edges \citep{huson1999disk, mossel2007distorted, gronau2012fast}. \citet{gronau2012fast} introduce the concept of \emph{adaptive fast converging algorithm} to refer to those that can reconstruct all edges  of length greater than any given threshold from sequence of polynomial length and construct the first algorithm with that property.

In our case, from Theorem \ref{thm:JukesCantor}, we deduce that if the length $k$ of the observed sequences satisfies
\[
\sqrt{2}\left(  C_{n,\delta,\gamma,e,g} +  \frac{\beta^2}{4} \right)^{1/2} \left(   \frac{\log k}{\alpha_k \sqrt{k} } +   \alpha_k \right)^{1/2} \leq  t_0,
\]
 then any edge of length greater than $t_0$ can be reconstructed by the regularized estimator with high probability.

 This leads to the following theorem, which shows that the regularized estimator is also adaptive fast-converging.

\begin{Theorem}
For any $\delta>0$, $t_0>0$ and $0 < \nu < 1/4$, if the length of the observed sequence $k$ satisfies
\[
k \sim \left ( \left [  \frac{1}{t_0^2} \left(\log \frac{1}{\lambda_e} \right)^2 \left ( \log \frac{4 g \gamma}{\lambda_{e} \delta} \right ) n^4  + \frac{\beta^2(n,g,f)}{t_0^2} \right ]^{\frac{4}{1 - 4 \nu}} \right )
\]
then the regularized estimator correctly reconstruct all edges of length greater than $t_0$ in the true tree $(\tau^*, q^*)$ with probability greater than $1-\delta$.
\label{thm:adaptive}
\end{Theorem}

In other words, if the model parameters (including $e$ and $g$) and the level of confidence $\delta$ are fixed, then the sample complexity to recover all edges of length greater than $t_0$ is

\[
k =  \mc{O} \left(\frac{n^{ \mc{O}(g)}}{(t_0f)^{8/(1-4\nu)}}\right)
\]
for any $0 < \nu < 1/4$.

\section{Application to yeast gene-tree reconstruction}
\label{sec:application}

In this section, we illustrate the applicability of our method in practical contexts by using the newly constructed estimator to analyze 106 widely distributed orthologous genes from the genomes of eight yeast species.
This data set was originally studied in \cite{rokas2003genome} and later incorporated to the \texttt{R} package \texttt{phangorn} \citep{phangorn} as a standard data set for phylogenetic analyses.

The previous results from \cite{rokas2003genome} indicate that data sets obtained from a single gene or by concatenating a small number of genes have a significant probability of supporting conflicting topologies.
By contrast, analyses of the entire data set of concatenated genes (by both maximum likelihood and maximum parsimony) yield a single, fully resolved species tree with high support.
Comparable results were obtained with a concatenation of a minimum of (randomly chosen) 20 genes; substantially more genes than commonly used.
One possible explanation for the success of using randomly chosen concatenated sequences to recover the species tree and the failure of the use of single-gene to do so is that there might not be sufficient information in single genes for maximum likelihood (or maximum parsimony) methods to reliably recover the correct tree.

The analyses from previous sections suggest a regularized maximum likelihood estimator of the form
\[
(\hat \tau_k, \hat q_k) :=  \argmin_{(\tau, q) \in \mathcal{T}}{- \frac{1}{k}\ell_k(\tau,q) +  C \frac{1}{k^{1/4}}R(\tau,q)}
\]
where $R(\tau,q) = d((\tau,q),(\tau^{\circ},q^{\circ}))^2$  and $d(s,s')$ denotes the BHV distance between the trees $s$ and $s'$ on the tree space.

To obtain the species tree $(\tau^{\circ}, q^{\circ})$,  we concatenate all genes in the data set and infer the species tree by maximum likelihood using package \texttt{phangorn} \citep{phangorn}. 
As explained in  \cite{rokas2003genome}, this tree receives strong and stable support from many standard reconstruction methods.
We will use the sequence data for the YKL120W gene as an example, noting that the reconstructed maximum likelihood tree is different from the species tree for that gene.
We note that analyses of other genes in the data set can also be obtained in a similar manner, but the incongruence between the ML tree and the species tree is of greater interest.
We compute the BHV distance using the \texttt{R} function \texttt{dist.multiPhylo} in \texttt{distory} package \citep{distory} and modify the likelihood computation sub-routine of the \texttt{phangorn} package to incorporate the penalty term into the optimization procedure.

While our theoretical results help provide insights about the convergence and asymptotic properties of the regularized estimator, they offer little about how the estimator performs when the tuning parameter is chosen in a data-dependent and stochastic way, or what is a proper procedure to tune such parameter.
There are several possible techniques for choosing the penalty parameter; however, the recommended (and default) technique for selecting $C$ in regularization problems is to choose $C$ such that the estimator minimizes the cross-validation estimate of the risk \citep{zhang1993model,homrighausen2013lasso}.
In our example, we will follow this standard procedure to select the constant $C$ for the penalty.

\begin{figure}[h]
\centering
  \includegraphics[width=0.6\linewidth]{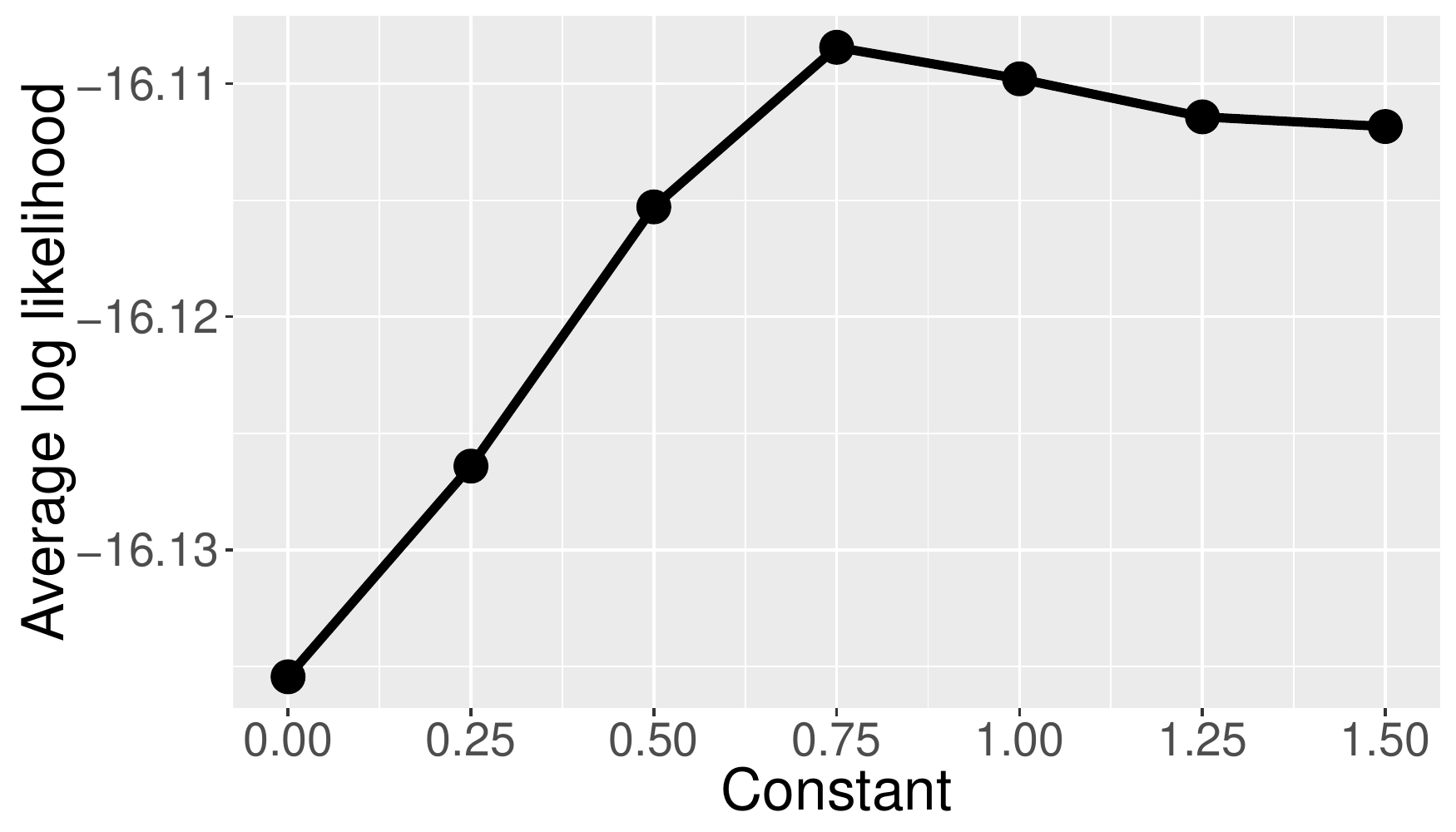}
  \caption{Average log likelihood from 3-fold cross validation. }
  \label{fig:loglike}
\end{figure}

The procedure for choosing $C$ can be summarized as follows.
We first randomly divide the sequence data into three parts, consider each part as validation data once, and use the remaining two parts to reconstruct the tree with various values of $C = 0, 0.25, 0.5, 0.75, 1, 1.25, 1.5$. Here, $C=0$ corresponds to maximum likelihood estimator.
This process is repeated $1000$ times and the average log-likelihood per site for each tree on the validation data is used to choose the constant $C$.
The result is summarized in Figure \ref{fig:loglike}. 
We observe that large values of $C$ (e.g.\ $1, 1.25, 1.5$) are preferred over small values ($0, 0.25, 0.5$), which indicates that the regularized method has better performance compared to the standard ML approach.
However, over-penalizing the distance to the species tree does not lead to better estimates.

Using the value of $C=0.75$ that maximizes the average log-likelihood, we reconstruct the regularized maximum likelihood tree, which is visualized in Figure \ref{fig:genetree}.
This gene tree has the same topology as the species tree, while the gene tree reconstructed by maximum likelihood method has a different topology (Figure \ref{fig:genetreeMLE}).
This result illustrates that while there might not be enough information from the (short) sequence data to reliably estimate the correct tree, additional knowledge from the species tree may help produce a better estimate (in terms of likelihood value on the validation data).
We note that this is consistent with the findings in \cite{rokas2003genome}, whereas adding more data from a different source improves the estimator.
In our case, the additional information from other genes is compressed in the guiding species tree and has the same effect.


\begin{figure}[!h]
\subfigure[Regularized method]{
\includegraphics[width=0.4\textwidth]{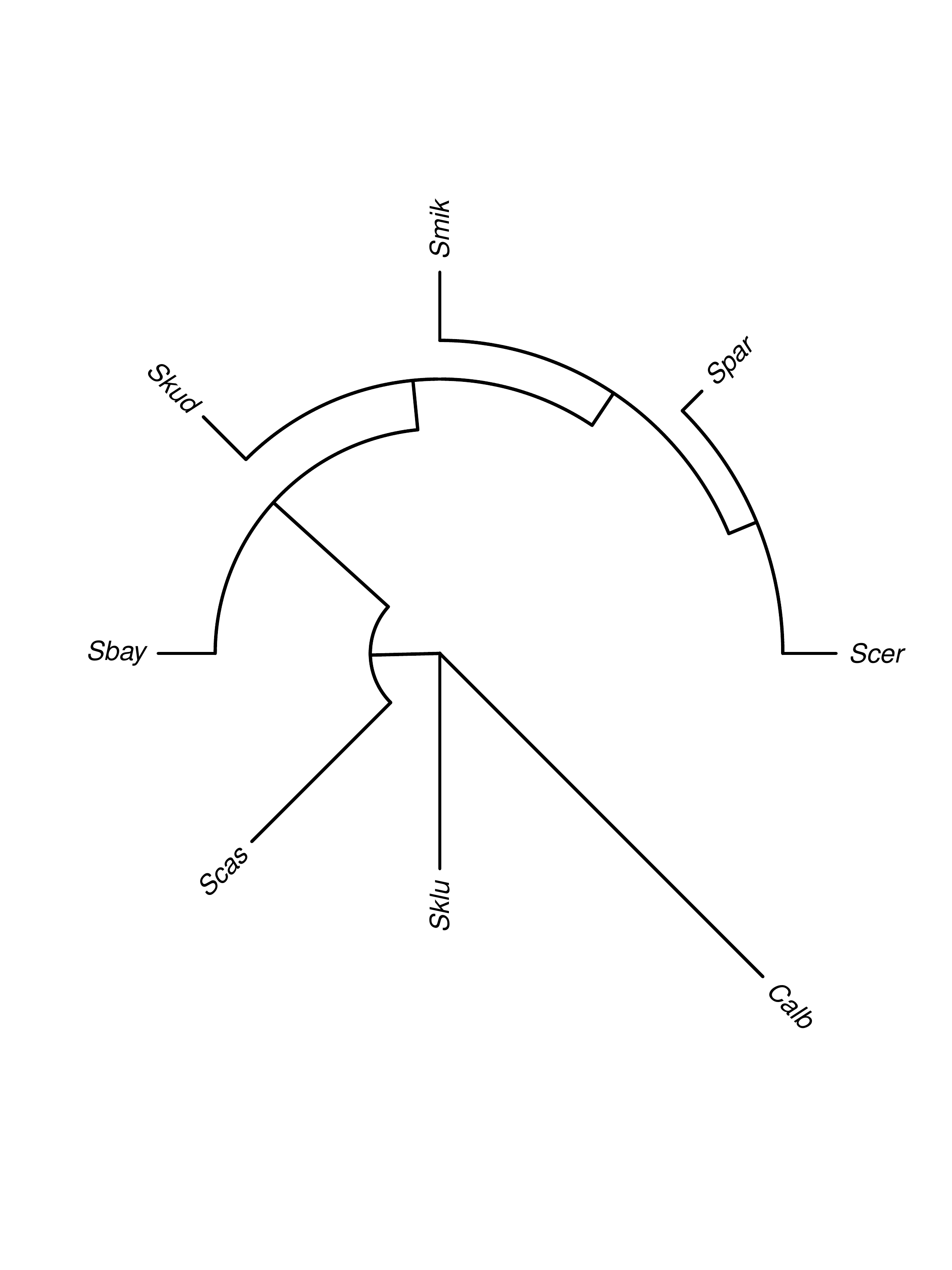}
\label{fig:genetree}
}
\subfigure[MLE method]{
\includegraphics[width=0.4\textwidth]{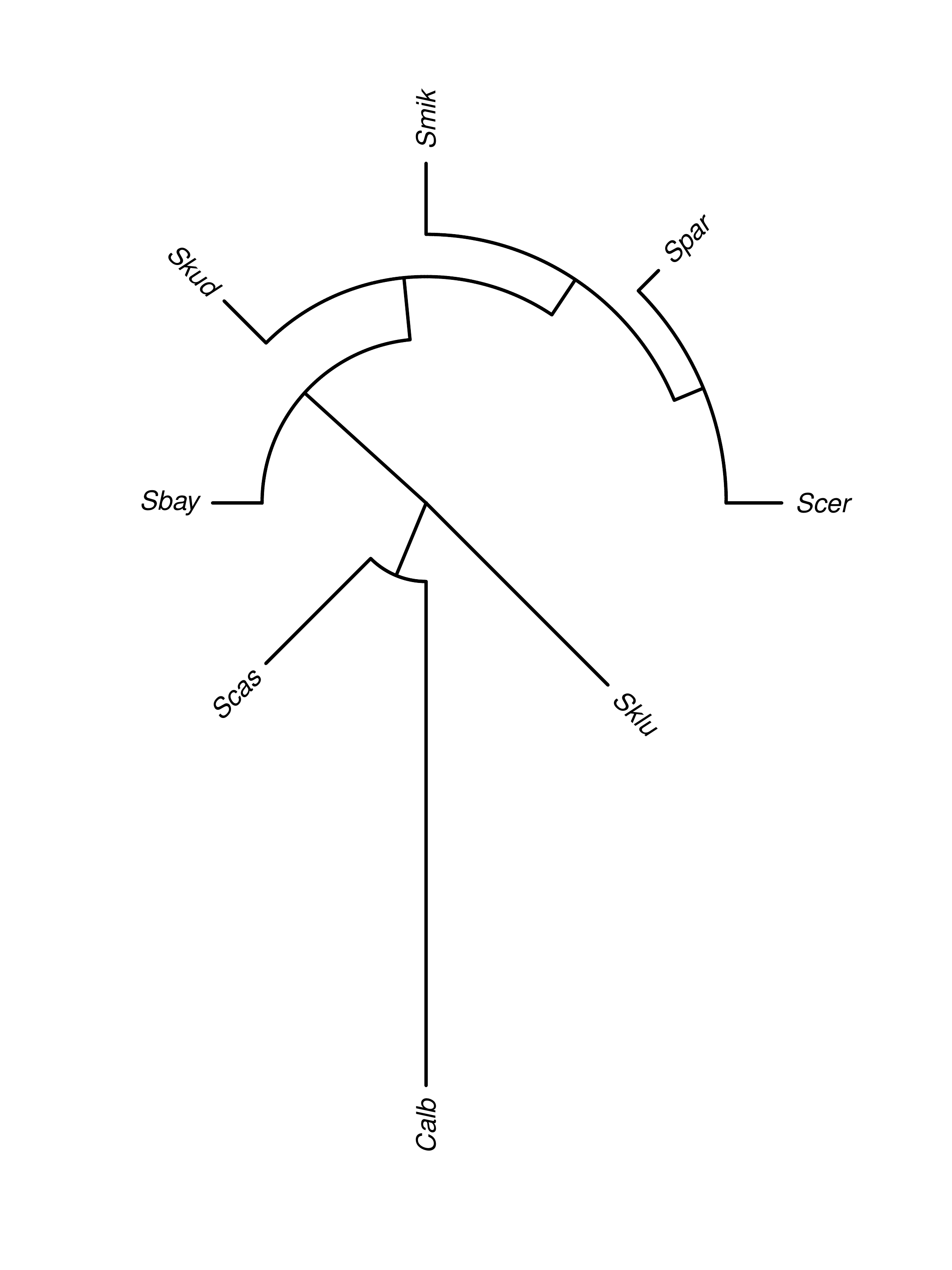}
\label{fig:genetreeMLE}
}
\caption{Reconstructed gene tree for the YKL120W gene.}
\end{figure}


\section{Discussion}
\label{sec:discussion}
In this paper, we propose a tree reconstruction method via regularization using the BHV geodesic distance and analyze its asymptotic properties.
We prove that the regularized maximum likelihood estimator is consistent and derive global error estimates of the estimator in various settings.
Using these estimates, we are able to provide an upper bound on the sample complexity of the estimator to recover the true topology.

\citet{steel2002inverting, steel2009inverting} provide a bound on the sample complexity of the maximum likelihood estimator, which requires that sequence length $k$ increase exponentially in the number of species $n$ in order to reconstruct the true topology.
With the same requirement, \citet{atteson1997performance} proves that the neighbor-joining method can recover the true topology for the Neyman $2$-state model on a binary trees for which the transition probability on all edges is bounded from below and above.
With the same assumptions in \citet{atteson1997performance}, \citet{erdos1999few} prove that the Dyadic Closure Method only needs $k$ to grow polynomially with respect to the number of species $n$ in order to reconstruct the true topology.
The work of \cite{huson1999disk, mossel2007distorted, daskalakis2011phylogenies, gronau2012fast} extend the result in \citet{erdos1999few} to construct various fast converging algorithms for tree reconstruction.
Recently, \citet{roch2015phase} prove that for $r$-state symmetric models, the sample complexity for MLE is $\mc{O}(\log n)$ if all the branch lengths are less than $\log \sqrt{2}$.
Otherwise, the sample complexity is $\text{poly}(n)$. \citet{mossel2003impossibility} gives a lower bound of order $\mc{O}(n^{\mu g /\log 2-1-o(1)})$ for the sample complexity of any tree reconstruction method on binary trees when $g> \log \sqrt{2}$.

Theorem \ref{thm:complexity} proves that the regularized maximum likelihood estimator recovers the true topology if $k = \text{poly}(n)$, which is the best possible sample complexity of other methods.
However, unlike most of the previous algorithms that perform poorly on trees with indistinguishable/very short edges, the regularized estimator can reconstruct all edges  of length greater than any given threshold from sequences of polynomial length.
Moreover, when information about the model configuration is available, we can derive a  conservative confidence region for the estimator and use that to assess the support of tree splits from data.
We also do not assume a lower bound for internal edges as previous work has done.

The uniform bounds on the convergence of the empirical likelihood, provided by Lemma $\ref{lem:unifboundg0}$ and Lemma $\ref{lem:ultiunifbound}$, are of independent interest.
These results allows us to consider the space of phylogenies as a continuous object instead of discretizing the branch lengths as has been done in previous analyses of sample complexity for tree reconstruction \citep[e.g.][]{roch2015phase}.
In a similar manner, we only assume that the lower bound for pendant edges is known and do not impose this condition on the inner edges.
This minor relaxation allows smooth transition between tree topologies and enables the uses of analytical methods to analyze phylogenetic algorithms.

There are several venues for improvement. First, in our current setting, we still assume that the upper bound on the branch lengths are known a priori.
This condition can be relaxed by extending Lemma $\ref{lem:ultiunifbound}$ to the case of unbounded edges and by highlighting the fact that the regularization term already imposes some implicit constraints on the branch lengths.
Secondly, we currently use the BHV geodesic distance as the penalty for regularization. While the geodesic distance is the most natural distance on the BHV space and is strongly convex, this distance cannot be computed in linear time.
We would like to extend the result to other (computationally cheaper) distances on tree space such as those proposed in \citet{amenta2007approximating}, or to other non-strongly convex penalties.
Relating to the species tree/gene tree problem, it is also worthwhile to consider penalty metrics based on explicit models of discordance-generating events, such as the duplication-loss cost investigated in \citet{wu2013treefix}.

\section*{Acknowledgments}

This work was partially support by the National Institutes of Health (R01 AI107034 and U54 GM111274) and the National Science Foundation (DMS 1223057, DMS 1264153, and CISE 1564137).
The research of Frederick Matsen was supported in part by a Faculty Scholar grant from the Howard Hughes Medical Institute and the Simons Foundation.

\bibliographystyle{chicago}
\bibliography{biblio}


\end{document}